\documentclass[sigconf,screen,pdfa]{acmart}
\copyrightyear{2024}
\acmYear{2024}
\setcopyright{rightsretained}
\acmConference[SPAA '24]{Proceedings of the 36th ACM Symposium on Parallelism in Algorithms and Architectures}{June 17--21, 2024}{Nantes, France}
\acmBooktitle{Proceedings of the 36th ACM Symposium on Parallelism in Algorithms and Architectures (SPAA '24), June 17--21, 2024, Nantes, France}\acmDOI{10.1145/3626183.3659958}
\acmISBN{979-8-4007-0416-1/24/06}

\makeatletter
\gdef\@copyrightpermission{
\begin{minipage}{0.3\columnwidth}
 \href{https://creativecommons.org/licenses/by/4.0/}{\includegraphics[width=0.90\textwidth]{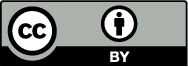}}
\end{minipage}\hfill
\begin{minipage}{0.7\columnwidth}
 \href{https://creativecommons.org/licenses/by/4.0/}{This work is licensed under a Creative Commons Attribution International 4.0 License.}
\end{minipage}
\vspace{5pt}
}
\makeatother

\usepackage{booktabs}   %% For formal tables:
\usepackage{subcaption} %% For complex figures with subfigures/subcaptions
\usepackage{caption}
\usepackage{colortbl}
\usepackage{bigstrut}
\usepackage{microtype}
\usepackage{tabularx}
\usepackage{multirow}
\usepackage{multicol}
\usepackage{rotating}
\usepackage{lipsum}
\usepackage{balance}
\usepackage{mfirstuc}%\captalisewords
\usepackage{titlecaps}%\titlecap
\usepackage{listings}
\usepackage{float}
\usepackage[rightcaption]{sidecap}
\usepackage{xcolor}
\usepackage[font=small,labelfont=bf,skip=2pt]{caption}
\usepackage[a-2b]{pdfx}

\def\fullversion{}
\usepackage{etoolbox}
\newcommand{\ifconference}[1]{{{\ifx\fullversion\undefined{#1}\fi}\xspace}}
\newcommand{\iffullversion}[1]{{{\ifx\conference\undefined{#1}\fi}\xspace}}

\usepackage{graphicx}  % pictures and figures
\usepackage{lipsum}  % random paragraphs
\newcommand{\hide}[1]{} % hide
\usepackage{xspace}
\usepackage{textcomp}
\usepackage{comment} % \begin{comment} ... \end{comment}
\usepackage{verbatim}
\usepackage{url}

\newcommand{\mtext}[1]{{\mbox{{#1}}}} % text in math mode
\newcommand{\defn}[1]{\emph{\textbf{#1}}} % definition style
\newcommand{\emp}[1]{{\boldmath\emph{\textbf{#1}}\unboldmath}} % highlight

\newtheorem{theorem}{Theorem}[section]
\newtheorem{lemma}{Lemma}[section]
\newtheorem{corollary}{Corollary}[theorem]

\let \originalleft \left
\let\originalright\right
\renewcommand{\left}{\mathopen{}\mathclose\bgroup\originalleft}
\renewcommand{\right}{\aftergroup\egroup\originalright}

\usepackage{scalerel} % stretch a formula

\newtheoremstyle{exampstyle}
{.5em} % Space above
{.5em} % Space below
{\it} % Body font
{.5em} % Indent amount
{\sc} % Theorem head font
{.} % Punctuation after theorem head
{.5em} % Space after theorem head
{} % Theorem head spec (can be left empty, meaning `normal')
\theoremstyle{exampstyle} 
\theoremstyle{exampstyle} 
\theoremstyle{exampstyle} 
\theoremstyle{exampstyle} 
\theoremstyle{exampstyle} \newtheorem{compactfact}{Fact}[section]

\makeatletter
\renewenvironment{proof}[1][\proofname]{\par
\vspace{-2\topsep}% remove the space after the theorem
\pushQED{\qed}%
\normalfont
\topsep0pt \partopsep0pt % no space before
\trivlist
\item[\hskip\labelsep
      \itshape
  #1\@addpunct{.}]\ignorespaces
}{%
\popQED\endtrivlist\@endpefalse
}

\newcommand{\whp}[1]{\emph{whp}}

\DeclareMathOperator*{\argmin}{\arg\!\min}
\DeclareMathOperator*{\polylog}{polylog}

\usepackage{pifont}
\usepackage[shortlabels]{enumitem}

\setlist{topsep=0.3em,itemsep=0.2em,parsep=0.1em,leftmargin=*}
\usepackage{float}
\usepackage[font={small},aboveskip=0.3em, belowskip=0.2em]{caption}

\usepackage[labelfont=bf,list=true,skip=0em]{subcaption}
\usepackage[rightcaption]{sidecap}

\usepackage{wrapfig}

\usepackage{array}% for extended column definitions
\newcolumntype{L}[1]{>{\raggedright\let\newline\\\arraybackslash\hspace{0pt}}m{#1}}
\newcolumntype{C}[1]{>{\centering\let\newline\\\arraybackslash\hspace{0pt}}m{#1}}
\newcolumntype{R}[1]{>{\raggedleft\let\newline\\\arraybackslash\hspace{0pt}}m{#1}}
\newcolumntype{B}{>{\bf}c}
\usepackage{rotating}

\usepackage{booktabs} % provides toprule, bottomrule, midrule, cmidrule, etc.
\usepackage{multicol,multirow}
\usepackage{longtable} % provides long table
\usepackage{supertabular} % similar to long table, allowing tables to take more than one page
\usepackage{colortbl}
\usepackage{bigstrut}
\newcommand{\myparagraph}[1]{\noindent\emp{#1} \ }

\usepackage[ruled,lined,linesnumbered,noend]{algorithm2e}
\usepackage[noend]{algpseudocode}

\makeatletter
\patchcmd{\@algocf@start}% <cmd>
  {-1.5em}% <search>
  {0pt}% <replace>
  {}{}% <success><failure>
\newcommand{\nosemic}{\renewcommand{\@endalgocfline}{\relax}}% Drop semi-colon ;
\newcommand{\dosemic}{\renewcommand{\@endalgocfline}{\algocf@endline}}% Reinstate semi-colon ;
\SetSideCommentLeft
\SetKwInput{notations}{Notations}
\SetKwInput{notes}{Notes}
\SetKwInput{maintains}{Maintains}

\SetKwProg{myfunc}{Function}{}{}
\SetKwFor{parForEach}{ParallelForEach}{do}{endfor}
\SetKwFor{Justrepeat}{Repeat}{}{}

\SetKw{MIN}{min}
\SetKw{MAX}{max}
\SetKw{OR}{or}
\SetKw{AND}{and}

\SetCommentSty{mycommfont}

\usepackage{mdframed}
\mdfdefinestyle{mystyle}{innertopmargin=2pt,innerbottommargin=2pt,skipabove=2pt,skipbelow=2pt}%leftmargin=0,topmargin=
\mdfdefinestyle{densestyle}{linecolor=framelinecolor,innertopmargin=0,innerbottommargin=0,leftmargin=0,rightmargin=0,backgroundcolor=gray!20}
\mdfdefinestyle{compactcode}{linecolor=framelinecolor,innertopmargin=1pt,innerbottommargin=1pt,backgroundcolor=gray!20,skipabove=0pt,skipbelow=0pt,leftmargin=0,rightmargin=0}

\usepackage{framed}

\usepackage{listings}

\newdimen\zzsize
\newdimen\kwsize
\newcommand{\basicstyle}{\fontsize{\zzsize}{1\zzsize}\ttfamily}
\newcommand{\keywordstyle}{\fontsize{\kwsize}{1\kwsize}\ttfamily\bf}

\newdimen\zzlstwidth
\usepackage{cleveref}
\crefname{section}{Sec.}{Sec.}
\crefname{theorem}{Thm.}{Thm.}
\crefname{lemma}{Lemma}{Lemma}
\crefname{corollary}{Cor.}{Cor.}
\crefname{table}{Tab.}{Tab.}
\crefname{algorithm}{Alg.}{Alg.}
\crefname{figure}{Fig.}{Fig.}
\crefname{fact}{Fact}{Fact}
\Crefname{table}{Tab.}{Tab.}
\crefname{problem}{Problem}{Problem}
\crefname{line}{Line}{Line}

\usepackage{tikz} % draw geometric objects

\hide{
\binoppenalty=700
\brokenpenalty=0 %100
\clubpenalty=0   %150
\displaywidowpenalty=0   %50
\exhyphenpenalty=50
\floatingpenalty=20000
\hyphenpenalty=50
\interlinepenalty=0
\linepenalty=10
\postdisplaypenalty=0
\predisplaypenalty=0 %10000
\relpenalty=500
\widowpenalty=0  %150
}

\newcommand{\naive}{na\"ive} 
\newcommand{\tinyskip}{\vspace{.5em}} 
\newcommand{\yan}[1]{{\color{violet}{\bf Yan:} #1}}

\newcommand{\best}{\mathit{best}}
\newcommand{\cordon}{\mathit{cordon}}

\newcommand{\now}{\mathit{now}}
\newcommand{\nxt}{\mathit{cordon}}

\newcommand{\stt}[1]{{#1}}
\newcommand{\pred}{\mathcal{P}}
\newcommand{\edge}[2]{#1 \rightarrow #2} 
\newcommand{\dep}{\mathit{depth}}
\newcommand{\ed}{\hat{d}}
\newcommand{\ds}{d^*}
\newcommand{\eff}[1]{\hat{#1}}
\newcommand{\ff}{\mathcal{F}}

\newcommand{\jl}{j_l}
\newcommand{\jr}{j_r}
\newcommand{\il}{i_l}
\newcommand{\ir}{i_r}
\newcommand{\im}{i_m}
\newcommand{\jm}{j_m}

\newcommand{\ffsize}{h}

\newcommand{\gammalws}{\Gamma_{\mathit{lws}}}
\newcommand{\gammagap}{\Gamma_{\mathit{gap}}}

\newcommand{\ouralgo}{Cordon Algorithm}
\newcommand{\sentinel}{sentinel}

\newcommand{\treelws}{Tree-GLWS}

\newcommand{\totarrows}{L} 

\newcommand{\vall}{\mathit{vall}}
\newcommand{\forest}{\mathit{forest}}

\newcommand{\old}{\mathit{old}}
\newcommand{\new}{\mathit{new}} 

\newcommand{\funcfont}[1]{{\textsf{#1}}}

\begin{document}

\title{
%Parallel Dynamic Programming Can Be (Nearly) Work-Efficient
Parallel and (Nearly) Work-Efficient Dynamic Programming
%(Nearly) Work-Efficient Parallel Dynamic Programming
}


% \affiliation{%
%   \city{Ohio}

\author{Xiangyun Ding}
\affiliation{
  \institution{University of California, Riverside}
  \country{}
}
\email{xding047@ucr.edu}

\author{Yan Gu}
\affiliation{
  \institution{University of California, Riverside}
  \country{}
}
\email{ygu@cs.ucr.edu}

\author{Yihan Sun}
\affiliation{
  \institution{University of California, Riverside}
  \country{}
}
\email{yihans@cs.ucr.edu}

\renewcommand{\shortauthors}{Xiangyun Ding, Yan Gu, and Yihan Sun}

\begin{abstract}

%The idea of dynamic programming (DP), since proposed by
%Richard Bellman in the 1950s, has been extensively used in
%algorithm design, and is one of the most important algorithmic
%techniques. 
The idea of dynamic programming (DP), proposed by
Bellman in the 1950s, is one of the most important algorithmic
techniques. 
%While being fundamental and relatively simple in the sequential setting,
%many DP problems are hard to 
However, in parallel, many fundamental and sequentially simple problems become
more challenging, and open to a (nearly) work-efficient solution (i.e.,
the work is off by at most a polylogarithmic factor over the best sequential solution). 
In fact, sequential DP algorithms employ many advanced optimizations such as
decision monotonicity or special data structures, and achieve
better work than straightforward solutions. 
Many such optimizations are inherently sequential, 
which creates extra challenges for a parallel algorithm 
to achieve the same work bound. 

The goal of this paper
is to achieve (nearly) work-efficient parallel
DP algorithms by parallelizing classic, highly-optimized and
practical sequential algorithms.
We show a general framework called the \emph{\ouralgo{}} for parallel DP algorithms,
and use it to solve several classic problems. 
Our selection of problems includes Longest Increasing Subsequence (LIS), sparse Longest Common Subsequence (LCS), 
convex/concave generalized Least Weight Subsequence (LWS), Optimal Alphabetic Tree (OAT), and more. 
%They can take advantage of classic optimizations (e.g., decision monotonicity). 
We show how the \ouralgo{} can be used to achieve the same level of optimization as the sequential algorithms,
and achieve good parallelism. 
Many of our algorithms are conceptually simple, and we show some experimental results as proofs-of-concept. 

\end{abstract}

%     <concept>
%         <concept_id>10003752.10003809.10003635</concept_id>
%         <concept_desc>Theory of computation~Graph algorithms analysis</concept_desc>
%         <concept_significance>500</concept_significance>
%         </concept>
%     <concept>
%         <concept_id>10003752.10003809.10010170</concept_id>
%         <concept_desc>Theory of computation~Parallel algorithms</concept_desc>
%         <concept_significance>500</concept_significance>
%         </concept>
%     <concept>
%         <concept_id>10003752.10003809.10010170.10010171</concept_id>
%         <concept_desc>Theory of computation~Shared memory algorithms</concept_desc>
%         <concept_significance>500</concept_significance>
%         </concept>
%   </ccs2012>

% \ccsdesc[500]{Theory of computation~Graph algorithms analysis}
% \ccsdesc[500]{Theory of computation~Shared memory algorithms}



\renewcommand\footnotetextcopyrightpermission[1]{} % This line removes the footnote about the conference and year.
\fancyhead{} % This line removes the page headers about the conference and authors.

%%%%%%%%%%%%%%%%%%%%%%%%%%%%%%%%%%%%%%%%%%%%
%% Display spacing
%%%%%%%%%%%%%%%%%%%%%%%%%%%%%%%%%%%%%%%%%%%%
% Put this after \begin{document}
\setlength\abovedisplayskip{0.2em}
\setlength\belowdisplayskip{0.2em}
\setlength\abovedisplayshortskip{0.2em}
\setlength\belowdisplayshortskip{0.2em}

\maketitle

\section{Introduction}\label{sec:intro}

The idea of dynamic programming (DP), since proposed by Richard Bellman in the 1950s~\cite{bellman1954theory}, has been extensively used in algorithm design, 
%In a high level, a DP algorithm defines a set of \emp{states} and their transitions (specified by a \emp{DP recurrence}).
%The algorithm will compute a \emp{DP value} for each state based on the recurrence.
%In this paper, we use integers to denote the states, and $D[i]$ the DP value of state $i$. 
and is one of the most important algorithmic techniques. 
It is covered in classic textbooks and basic algorithm classes, and is widely used in research and industry. %papers on various advanced optimizations~\cite{dreyfus1977art,knuth1997art}. %\xiangyun{need more citations}.
The goal of this paper is to achieve \emp{(nearly) work-efficient} (defined below) and \emp{parallel} DP algorithms 
based on parallelizing classic, highly-optimized and practical sequential algorithms.

At a high level, a DP algorithm computes the \emp{DP values} for a set of \emp{states} (labeled by integers) by a \emp{recurrence}.
%More concepts about DP are introduced in \cref{sec:prelim}. 
%We use $D[i]$ to denote the DP value of state $i$. 
%Let $D[i]$ be the DP value of state $i$.
The recurrence specifies a set of \emp{transitions} from state $j$ to state $i$, i.e., how $D[j]$ can be used to compute $D[i]$\footnote{
%More generally, a transition may compute $D[i]$ from
%multiple other states. We keep it simple because all algorithms in this paper only requires one other state $j$ in the transition to compute $D[i]$.
More generally, a transition may compute $D[i]$ from multiple other states. 
All algorithms in this paper only requires one state $j$ in the transition to compute $D[i]$.
}.
%\xiangyun{This is only one type of DP? $D[i]$ may depends on several other states (like OBST).}
We call $j$ a \emp{decision} at $i$. 
$D[i]$ is then computed by taking the best (minimum or maximum) among all decisions,
which we call the \emp{best decision} at $i$.
%We call the decision that gives the final value of $D[i]$ the \emp{best decision} at $i$.
Throughout the paper, we will use $i^*$ to denote the best decision of state $i$. 
We introduce more concepts about DP in \cref{sec:prelim}. 

\hide{
At a high level, a DP algorithm solves an optimization
problem by breaking it down to simpler subproblems, memoizing the answers to the subproblems, and using them to find the answer to the original problem. 
The subproblems are usually indexed by integers, referred to \emp{states}. %of the original problem.
The DP value of a state is determined either by a \emp{boundary} condition (i.e., initial values),
or from other states, specified by a \emp{DP recurrence}. 
The recurrence specifies a set of \emp{transitions} from state $j$ to state $i$, i.e., how $D[j]$ can be used to compute $D[i]$.
We call $j$ a \emp{decision} at $i$. 
$D[i]$ is then computed by examining all decisions at state $i$, and taking the best (minimum or maximum) among them. 
We call the decision that gives the final value of $D[i]$ the \emp{best decision} at $i$, and denote it as $i^*$. 
}

One can view the states and transitions as a directed acyclic graph (DAG), which we refer to as a \emp{DP DAG}. 
In this DAG, each vertex is a state, and an edge $j$ to $i$ denotes a transition from $j$ to $i$. 
Since such an edge indicates that computing $D[i]$ logically requires $D[j]$, 
we also call it a \emp{dependency}, and say $i$ \emp{depends on}~$j$. 
Sequentially, we can compute all states based on a topological ordering. 
For simplicity, we always assume that the (integer) order of the states is a valid topological ordering.
% (i.e., the order given by running topological sort on the DP DAG). 

Unfortunately, many DP algorithms (even some simple ones sequentially) are hard to parallelize,
and are especially hard to achieve work-efficiency (the work asymptotically matches the best sequential algorithm) or even near work-efficiency 
(off by a polylogarithmic factor). 
We note that on today's multicore machines with tens to hundreds of processors, 
achieving low work is one of the most crucial objectives for designing \emph{practical} parallel algorithms. 
%Many reasons contribute to such difficulty. 
One particularly intriguing and somewhat ironic challenge for achieving work-efficient parallel DP algorithms is that, 
sequential algorithms are extremely well-optimized. 
In many cases, an optimized DP algorithm does not need to process all edges (transitions) in the DP DAG; 
some even do not need to process all vertices (states).
We review the literature at the end of this section. 
In fact, almost all textbook DP solutions can be optimized to achieve lower work than the straightforward solution 
that directly computes the DP values of all states based on the recurrence. 
Such examples include longest increasing subsequence (LIS), (sparse) longest common subsequence (LCS), 
(convex/concave) least weight subsequence (LWS), and many others discussed in this paper. 

One typical DP optimization that is both theoretically elegant and practically useful 
is \emp{decision monotonicity (DM)}.
\hide{At a high level, DM indicates that,
if a state $i$ has best decision $i^*$,
then another state $j>i$ must have its best decision $j^*\ge i^*$, called \emph{convex} case, 
or the \emph{concave} case\footnote{The definitions of convexity and concavity are interchanged in some other papers.}
where $j^*\le i^*$ or $j^*\ge i$.}
At a high level, DM indicates that two states $i$ and $j>i$ must have their best decisions $j^*\ge i^*$, called the \emph{convex} case\footnote{The definitions of convexity and concavity are interchanged in some other papers.},
or the \emph{concave} case where either $j^*\le i^*$ or $j^*\ge i$ (see \cref{fig:dm}). 
%then in the \emph{convex} case, another state $j>i$ must have its best decision $j^*\ge i^*$, or in the \emph{concave} case, $j^*\le i^*$ or $j^*\ge i$.
%This paper focuses on the convex case, since it is widely applicable in the real world~\cite{eppstein1988speeding} (but our techniques can also apply to the concave case). 
%Convexity means that the partial order between any two states $i$ and $j$
%are the same as their best decisions $i^*$ and $j^*$. 
Hence, when finding the best decision for state~$j$,
one can narrow down the possible range of $j^*$ based on the known best decisions of previous states, 
and thus avoid processing all transitions. 
DM has been widely studied in the sequential setting (e.g.,~\cite{klawe1989simple,klawe1990almost,eppstein1992sparse2,galil1992dynamic,eppstein1990sequence,galil1989linear,klawe1989simple}),
and is also closely related to concepts such as 
quadrangle inequalities~\cite{yao1980efficient,yao1982speed} and Monge property~\cite{monge1781memoire}.
Sequentially, 
using DM saves a polynomial factor than the na\"ive DP algorithm in many applications~\cite{miller1988sequence,galil1989speeding,knuth1981breaking,aggarwal1987geometric,wilber1988concave,aggarwal1990applications,BG2020}. 
In the parallel setting, however, among the papers we know of~\cite{galil1994parallel,huang1994parallel,larmore1995constructing,rytter1988efficient,bradford1998efficient,czumaj1992parallel,chan1990finding},
most from the 90s, very few of them take advantage of DM to reduce work. 
Indeed, none of them are work-efficient, and most of them have a polynomial overhead, 
which limits their potential applicability on today's multicore machines. 
The only nearly work-efficient results~\cite{chan1990finding,bradford1998efficient} focus on the concave case of one specific problem.
%Hence, although these algorithms are intellectually interesting, they are mostly of theoretical interest consider the  today's parallel machines.
\begin{figure}
  \centering
      \centering
      \includegraphics[width=\columnwidth]{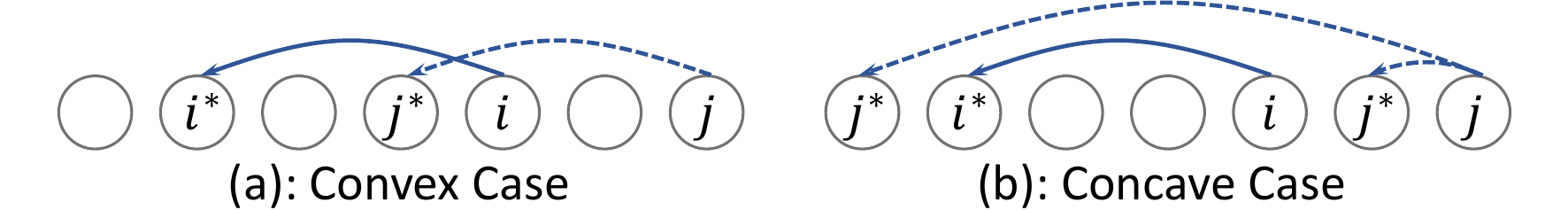}
  \caption{Convex and concave decision monotonicity.  
  (a). Convexity: for two states $i<j$, their best decisions satisfy $i^*\le j^*$.  
  (b). Concavity: for two states $i<j$, their best decisions satisfy either $j^*\le i^*$ or $j^*\ge i$.
  }
  \label{fig:dm}

\end{figure}

\hide{
In parallel, however, 
among the papers we know of concerning DP with DM optimizations (e.g., on edit distance~\cite{aalm}, xxx), %mainly focusing on edit distance (cite AALM and the other ones) and some other specific problems (cite Lawrence's papers).
%Unfortunately, these algorithms are also not work-efficient when comparing to the best known sequential algorithms using DM properties on these topics.
none of them are (even nearly) work-efficient when comparing to the sequential solutions using DM. 
}

The challenge of using DM in parallel lies in two aspects. 
First, sequentially we skip the transitions for state $i$ by observing the best decisions of all states before $i$.
When multiple states are processed in parallel, they cannot see each other's best decision, 
making it hard to skip the same set of ``useless'' transitions as in the sequential setting.
Second, many classic sequential DP algorithms with DM relies on efficient data structures such as monotonic queues, which are inherently sequential. 
Achieving the same work bound in parallel also requires careful redesign of the underlying data structures. 
%Achieving the same work bound in parallel requires careful redesign the parallel versions of these classic data structures. 

%In this paper, we study efficient parallel DP algorithms with optimizations. 
In this paper, we study parallel DP algorithms to achieve the same work as highly-optimized sequential algorithms. 
%For work-efficiency, our goal is to achieve the same level of sparsification as in the sequential algorithms.
Given a sequential algorithm $\Gamma$ with certain optimizations, 
our goal for work-efficiency is to process (asymptotically) the same number of transitions and states as in $\Gamma$. 
Regarding parallelism, we hope to achieve the best possible parallelism indicated by the transitions processed by $\Gamma$.
%specified by the algorithm DAG of $\Gamma$.
We formalize our goals in \cref{sec:framework}. 
%i.e., achieving $\tilde{O}(d)$ span, 
%where $d$ is the depth of the algorithm DAG of $\Gamma$, namely, the logical dependence length of the computations in $\Gamma$. 
Our solution is based on an algorithmic framework that generally applies to almost \emph{all} DP algorithms.
We call this framework the \emp{\ouralgo{}}. 
%At a high level, this algorithm specifies how to \emph{correctly} find all ready states. 
At a high level, our framework specifies how to \emph{correctly} identify a subset of states
that do not depend on each other and process them in parallel. 
%specifies how to \emph{correctly} find all ready states. 
We then present how to do so \emph{efficiently} for specific problems.
To achieve work-efficiency, our key ideas are two-fold. 
First, many of our algorithms use \emph{prefix-doubling} 
to bound the additional work on processing unnecessary states.
Second, we design \emph{new parallel data structures} to skip unnecessary transitions. 
%identify the ``useful'' transitions and states. 
%using parallel data structures for each studied problem. 

This paper studies general approaches for parallel DP, 
with a special focus on applying the non-trivial, effective optimizations found in the sequential context to parallel algorithms. 
We select classic DP problems and their optimized sequential solutions, and parallelize them using novel techniques. 
Our framework unifies one existing parallel LIS algorithm~\cite{gu2023parallel}, 
and provides new parallel algorithms for various problems such as sparse LCS,
convex/concave generalized LWS and GAP edit distance (GAP), and optimal alphabetic tree (OAT).
All the algorithms are (nearly) work-efficient with non-trivial parallelism. 
Among them, we highlight our contributions on parallelizing the DP algorithms with decision monotonicity. 
The core of our idea is a parallel algorithm for convex/concave generalized LWS. 
We apply it to other problems such as GAP and OAT, and achieve new theoretical results. 
For OAT, we partially solve the open problem in~\cite{larmore1993parallel} by providing 
a work-efficient algorithm with polylogarithmic span when the inputs are positive integers with word size $n^{\polylog(n)}$. 
We present our theoretical results in \cref{thm:1d:concave,thm:1d:convex,thm:gap,theorem:oat,thm:lcs,thm:lis}. 
We believe this is the first paper that achieves near work-efficiency in parallel for a class of DP algorithms with DM.

Although the main focus of this paper is to achieve low work in theory, 
an additional goal is to make the algorithms simple and practical. 
We implement two algorithms as proofs-of-concept (code available at~\cite{dpdpcode}). 
On $10^9$ input size, 
both of them outperform sequential solutions 
when the depth of the DP DAG is within $10^5$,
and achieve 20--30$\times$ speedup with smaller depth of the DP DAG. 
%Our code is released at~\cite{}. 

\ifconference{
Due to page limit, we provide the full version of the paper in~\cite{fullversion}, which contains more algorithmic details, proofs, and some additional algorithms that fit into our framework. 
}

\myparagraph{Related Work.}
Dynamic programming (DP) is one of the most studied topics in algorithm design.
The seminal survey~\cite{galil1992dynamic} by Galil and Park reviewed two types of optimization techniques sequentially, including decision monotonicity (e.g.,~\cite{klawe1989simple,klawe1990almost,wilber1988concave,eppstein1992sparse2,galil1992dynamic,eppstein1990sequence,galil1989linear,klawe1989simple}) and sparsity (e.g.,~\cite{eppstein1990sparse,eppstein1992sparse,hirschberg1977algorithms,hunt1977fast}).
This paper mainly focuses on parallelizing the sequential algorithms in this scope, and we will review the literature of each problem in the corresponding section.

DP is also widely studied in parallel.
%Many existing algorithms are optimized in different models, 
There exists rich literature on optimizing various goals in different models, 
such as span (time) in PRAM (e.g.~\cite{galil1994parallel,apostolico1990efficient,babu1997parallel,lu1994parallel,larmore1994optimal,atallah1989constructing,huang1994parallel,larmore1995constructing,rytter1988efficient,larmore1996parallel,larmore1993parallel}),
I/O cost in the external-memory/ideal-cache model~(e.g.,~\cite{chowdhury2006cache,chowdhury2008cache,chowdhury2010cache,chowdhury2016autogen,itzhaky2016deriving,tang2015cache}), rounds in the MPC model~\cite{im2017efficient,boroujeni2019improved,gupta2023fast,bateni2018massively}, or on the BSP model~\cite{krusche2009parallel,krusche2010new,alves2002parallel}. 
However, these papers either only considered the DP algorithms without the optimizations, 
or incur polynomial work overhead, except for \cite{chan1990finding,bradford1998efficient} for one specific problem.
Instead, our paper tries to parallelize the efficient and practical sequential algorithms while maintaining low work. 
Some other works try to parallelize certain types of DP algorithms or applications using DP (e.g.,~\cite{awan2020adept,aimone2019dynamic,blelloch2016parallel,weiss2019computation,javanmard2019toward,li2014parallel,ding2023efficient}).
Alternately, our work aims to provide a general approach to parallelize almost all DP algorithms.

\section{Model and Framework}
\label{sec:prelim}

%\myparagraph{Computational Model.}
We use the \defn{work-span model} in the classic multithreaded model with \defn{binary-forking}
\cite{arora2001thread, blelloch2020optimal, blumofe1999scheduling}.
We assume a set of threads that share the memory.
Each thread acts like a sequential RAM plus a fork instruction that forks two threads running in parallel.
When both threads finish, the original thread continues.
A parallel-for is simulated by forks in a logarithmic number of steps.
A computation can be viewed as a DAG.
The \defn{work} $W$ of a parallel algorithm is the total number of operations,
and the \defn{span (depth)} $S$ is the longest path in the DAG.
In practice, we can execute the computation with work $W$ and span $S$ using a randomized work-stealing scheduler~\cite{blumofe1999scheduling,gu2022analysis} in time $W/P+O(S)$
with $P$ processors with high probability.
A parallel algorithm is \emp{work-efficient}, if its work is $O(W)$, where $W$ is the work of the best known or the corresponding sequential algorithm,
and \emp{nearly work-efficient} if its work is $\tilde{O}(W)$. We use $\tilde{O}(\cdot)$ to hide $\mtext{polylog}(n)$ where $n$ is the input size. %We use $\log n$ as a shorthand for $1+\log_2(n+1)$. 

\hide{
\tinyskip

In the following, we start with presenting some basic concepts and optimizations of DP.
We then present our algorithmic framework, the \ouralgo{}, for parallel DP algorithms. 
We note that \ouralgo{} by itself \emph{does not} guarantee efficiency. 
Later in \cref{sec:xx} we present how to use more techniques to enable (near) work-efficiency for the
problems of interest. }

\subsection{Basic Concepts in Dynamic Programming}

A DP algorithm solves an optimization
problem by breaking it down to subproblems, memoizing the answers to the subproblems, and using them to find the answer to the original problem. 
The subproblems are usually indexed by integers, referred to as \emp{states}. %of the original problem.
With clear context, we directly use $i$ to refer to ``state~$i$''.
The DP value of a state is determined either by a \emp{boundary} condition (i.e., initial values),
or from other states, specified by a \emp{DP recurrence}. 
This paper studies recurrences in the following form:
\begin{equation}
\label{eqn:dp-recur}
  D[i]={\min/\max}~f_{i,j}(D[j])
\end{equation}
%In the recurrence, $D[i]$ can be computed from several other states (noted as a set $\pred(i)$), 
%called the \emp{decisions} at \stt{i}. 
%Here $\pred(i)$ is the set of decisions of state $i$.
%, which are also the predecessors of state $i$ in the DAG. 
where $j$ is a decision at $i$.
Function $f_{i,j}(\cdot)$ indicates how the DP value of state $j$ can be used to update state $i$.
%$D[i]$ chooses the best among all decisions. 
%If \stt{j} is a decision at \stt{i}, we say \stt{i} \emp{depend on} \stt{j}. 
%Let $\pred(i)=\{j\,|\,\text{\stt{i} depends on \stt{j}}\}$. 
%The DP recurrences are usually in the following form. 
%We say state $j^*$ is the \emp{best decision} of state~$i$, if $D[i]=f_{i,j}(D[j^*])$, i.e.,
%state $j$ gives the optimal solution for state~$i$. We use array $\best[i]=j^*$ to record the best decision for each state~$i$. 
The transitions (i.e., dependencies) among states form a DAG $G=(V,E)$ as introduced in \cref{sec:intro}. 
%where each vertex is a state, and a directed edge from $j$ to $i$ means that $i$ depends on $j$. 
%We use \boldmath$\dep(G)$\unboldmath{} to denote the depth of $G$. 
We use $\edge{j}{i}$ to denote an edge from $j$ to $i$ in the DAG, and use $\best[i]$ to denote the best decision of state $i$.
%Clearly, $\pred(i)$ contains all predecessors of \stt{$i$} in the DAG. 

\hide{In the recurrence, $D[i]$ can be computed from several other states (noted as a set $\pred(i)$), 
called the \emp{decisions} at \stt{i}, 
and chooses the best (min or max) one among them.
In our paper, we consider the scenario where each decision at $i$ only depend on one other state. 
If \stt{i}'s DP value can be computed from \stt{j}, we say \stt{i} \emp{depends on} \stt{j}. 
%In the recurrence, $D[i]$ is chosen from the best (min/max) among several other states (noted as a set $\pred(i)$), 
%Let $\pred(i)=\{j\,|\,\text{\stt{i} depends on \stt{j}}\}$. 
}

During a DP algorithm, we may maintain the DP value of a state and update it by the recurrence. 
We call the process of updating $D[i]$ by $D[j]$ a \emp{relaxation}, or say \stt{$j$} \emp{relaxes}~\stt{$i$}. 
A relaxation is successful if the DP value is updated to a \emp{better} value, i.e., a lower (higher) value if the objective is minimum (maximum).
%We say an algorithm \emp{evaluates} an edge $\edge{j}{i}$ or a state $i$ 
%If an algorithm uses \stt{$j$} to relax \stt{$i$}, we say it \emp{evaluates} the edge $\edge{j}{i}$ or the state $i$. 
We call the actual DP value of a state the \emp{true} or \emp{finalized} DP value to distinguish from the DP value being updated during the algorithm, 
which we call the \emp{tentative} DP value. 
%Understanding the \ouralgo{} needs some basic concepts.
We say a state $i$ is finalized if we can ensure that its true DP value has been computed, and tentative otherwise.
%We also call a DP value a \emph{finalized (tentative) DP value} if the state is finalized (tentative). 
Among all tentative states, we say a state is \emp{ready}, if all its decisions are finalized, and \emp{unready} otherwise. 
%The \emp{rank} of a state is its depth in the DP DAG. 
A \naive{} DP algorithm will process all transitions and states based on a topological ordering.
%We say a state $i$ is \emp{finalized} if we can ensure that its true DP value have been computed, and \emp{tentative} otherwise.
%We also call a DP value a \emph{finalized (tentative) DP value} if the state is finalized (tentative). 
%Among all tentative states, we say a state is \emp{ready}, if all its predecessors are finalized, and \emp{unready} otherwise. 

\myparagraph{DP Optimizations.} 
Instead of computing the recurrence straightforwardly, 
%many algorithms can {optimize} the computation by modifying the DAG based on some rules. 
%Namely, $G_{\Gamma}$ can be different from the DP DAG directly suggested by the recurrence. 
many algorithms can optimize the computation by skipping vertices and/or edges in the DAG to save work. 
We call such algorithms \emp{optimized DP algorithms}. 
%many \emp{optimizations} can be applied to the recurrence and 
%Such optimizations can possibly skip vertices and/or edges in the DAG to save work. 
%For example, an edge $\edge{j}{i}$ can be skipped if some information reveals that $j$ cannot be $\best[i]$.
%Sequentially, there exist many well-studied DP optimizations. 
%We carefully review them and categorized them into two groups.
%We discuss two well-adopted DP optimizations as examples below.
For example, given an input sequence $A[1..n]$, the optimized LIS algorithm~\cite{Knuth73vol3}
maintains a data structure to precisely find the best decision of each state in $O(\log n)$ cost, 
and only processes $O(n)$ transitions instead of $O(n^2)$ as suggested by the recurrence. 
Similarly, given two input sequences $A[1..n]$ and $B[1..m]$, 
the optimized LCS algorithm only needs to process all states $D[i,j]$ where $A[i]=B[j]$,
instead of $O(nm)$ states in the recurrence. 
Another typical optimization is decision monotonicity where the best decisions of previous states can narrow down the range for best decisions for later states, 
which skips transitions and saves work.
Typical examples include concave/convex generalized LWS (see \cref{sec:post-office}), OAT (see \cref{sec:oat}), GAP (see \cref{sec:gap}), etc. 
%We will apply our algorithm to all these examples discussed above. 
We will discuss all these algorithms in this paper.

\subsection{Parallelizing Sequential DP Algorithms}

We now discuss our goal to parallelize a sequential DP algorithm.
Our primary goal is to achieve (asymptotically) the same computation as an optimized sequential algorithm. 
%Specifically, for a sequential DP algorithm with optimizations, we want to 1) match the optimized work as the sequential algorithm and 2) achieve the best possible parallelism. 
%For an algorithm~$\Gamma$, we organize all states and dependencies evaluated by $\Gamma$ in a DAG, called a \emp{$\Gamma$-optimized DAG}, denoted as $G_{\Gamma}$. 
More formally, for a sequential algorithm~$\Gamma$ computing a recurrence $R$ with certain optimizations, 
we define the %\emp{optimized DAG} of $\Gamma$, 
\emp{$\Gamma$-optimized DAG}, 
denoted as $G_{\Gamma}=(V_{\Gamma},E_\Gamma)$, as follows. 
$G_{\Gamma}$ is the same as the DP DAG on $R$ with some edges highlighted: 
for all edges that are processed by $\Gamma$, 
we highlight them in $G_{\Gamma}$, and call them the \emp{effective edges}. 
We call the other edges \emph{normal edges}.
%All transitions in the recurrence $R$ will be included in $E_{\Gamma}$ and marked \emp{grey}. 
%For all edges that are evaluated by the optimized algorithm $\Gamma$, 
%we call them the \emp{effective edges}, and mark them \emp{red} in $G_\Gamma$. 
%For the edge set $E_{\Gamma}$, we also denote the subset of all red edges as $E^*_{\Gamma}$, 
%which are all the \emph{effective} edges that can be used to fully complete the computation. 
These effective edges can be used to fully complete the computation. 
%organize all states and dependencies evaluated by $\Gamma$ in a DAG, called a \emp{$\Gamma$-optimized DAG}, denoted as $G_{\Gamma}$. 
For a sequential DP algorithm $\Gamma$, we hope its faithful and best possible parallelization $\Lambda$ to:

\begin{itemize}
  \vspace{-0.08cm}
  \item process the same effective edges in $G_\Gamma$ and achieve the same work as $\Gamma$, and
  \vspace{-0.08cm} 
  \item have span proportional to the \emp{effective depth} of $G_\Gamma$, 
  defined as the largest number of effective edges in any path in $G_\Gamma$. 
\end{itemize}
\vspace{-0.08cm}
Namely, if our goal is to parallelize $\Gamma$ and perform the same computation, 
the parallel algorithm $\Lambda$ has to process all edges $\edge{j}{i}$ processed by $\Gamma$. 
This means \stt{$i$} can be finalized only after \stt{j} is finalized, 
%Logically this gives dependency between $j$ and $i$, 
so the span is related to the effective depth as defined above. 
%$\dep(G_{\Gamma})$. 
%We use $\dep(G_\Gamma)$ to denote the weight depth of $G_\Gamma$, where each red edge has a unit weight, and each grey edge has no weight. 
We use $\eff{E}_{\Gamma}$ to denote the set of effective edges,  
%Let the effective depth of a state $s$ be the largest number of effective edges on any path ending at $s$, denoted as $\dep^+(s)$. 
and $\ed(G_\Gamma)$ as the effective depth of $G_\Gamma$. 
%is the largest effective depth among all states.
%We hope the parallel algorithm to evaluate (asymptotically) the same effective edges as $\Gamma$.  
%We use $\dep^+(G_\Gamma)$ to denote effective depth of $G_\Gamma$,
%Note that one can still , but it may require to redesign the computation and break the dependency, 
%which makes the parallel algorithm no longer the same computation as the sequential algorithm $\Gamma$. 
%For the span, we usually allow a polylogarithmic overhead considering binary-forking. 
%To formalize our goals, we say define the following concepts. 
More formally, we define the following concepts.

We say a parallel algorithm $\Lambda$ is an \emp{optimal parallelization} of a sequential DP algorithm~$\Gamma$, if 
1) the set of edges processed in $\Lambda$ and $\Gamma$ satisfies $\eff{E}_{\Gamma}\subseteq \eff{E}_{\Lambda}$ and $|\eff{E}_{\Lambda}|=O(|\eff{E}_{\Gamma}|)$, 
2) the work of the two algorithms are asymptotically the same, and
3) the span of $\Lambda$ is $\tilde{O}(\ed(G_\Gamma))$. 
Namely, $\Lambda$ can process a superset of edges of that in~$\Gamma$, but not asymptotically more,
and its span is proportional to the effective depth of the $\Gamma$-optimized DAG.
%We allow for $\Lambda$ to evaluate slightly more edges since in many cases it is hard to evaluate exactly the same set of edges while achieving high parallelism, especially when the sequential algorithm has non-trivial optimizations. 

In addition, we define the \emp{perfect parallelization}.
In an omniscient version of $\Gamma$, we only need to process the best decisions based on their dependencies. 
We define the \emp{$\Gamma$-perfect DAG}, denoted as $G^*_\Gamma$, 
as subgraph of $G_\Gamma$ that only contains the best decision edges (both effective and normal edges).
%We define the \emph{perfect depth} of a state $s$ as follows.
%Starting from $s$, we will chase the best decision path in the DAG until we see a state initialized by the boundary condition.
%The number of effective edges on this path is the \emph{perfect depth} of $s$.
%The perfect depth of a DAG $G$, denoted as $\dep^*$ is the largest perfect depth among all states. 
We say an optimal parallelization $\Lambda$ of $\Gamma$ is also a \emph{perfect parallelization} 
if the span of $\Lambda$ is $\tilde{O}(\ed(G^*_{\Gamma}))$. 

%We also call $\ed(G_\Gamma)$ the \emph{effective depth} of algorithm $\Gamma$, and $\ed(G^*_\Gamma)$ the \emph{perfect depth} of $\Gamma$. 

%We also define $\ed(\Gamma)=\ed(G_\Gamma)$ as the effective depth

\hide{
In addition, we define the \emp{perfect parallelization} of a sequential algorithm $\Gamma$ as follows.
%We consider the DAG $G_\Pi=(V_\Pi, E_\Pi)$, which only consists of edges from $\best[i]$ to $i$ for a given recurrence,
%with possible shortcuts. 
%Consider an omniscient version of algorithm $\Gamma$, denoted as $\Gamma^+$, that knows the best decision of all states and only evaluates them in a proper order. 
In an omniscient version of $\Gamma$, we only need to evaluate the best decisions based on their dependencies. 
We define the \emph{perfect depth} of a state $s$ as follows.
Starting from $s$, we will chase the best decision path in the DAG until we see a state initialized by the boundary condition.
The number of effective edges on this path is the \emph{perfect depth} of $s$.
The perfect depth of a DAG $G$, denoted as $\dep^*$ is the largest perfect depth among all states.
An optimal parallelization $\Lambda$ of $\Gamma$ is also a \emph{perfect parallelization} 
if the the span of $\Lambda$ is $\tilde{O}(\dep^*(G_{\Gamma}))$. 
}
\hide{
We say a perfect parallelization $\Lambda$ is a \emp{perfect parallelization} of $\Gamma$, if 
%1) $\Lambda$ evaluates the same set of edges as, and have the same work as $\Gamma$, and
1) $\Lambda$ is an optimal parallelization of $\Gamma$, and
2) the span of $\Lambda$ is $\tilde{O}(\dep^*(G_{\Gamma}))$. 
}

While the definitions seem abstract, we will later show that they are intuitive for concrete problems. 
For example, in longest common subsequence (LCS), both $\ed(G_{\Gamma})$ and $\ed(G^*_{\Gamma})$ are the output LCS length $k$ 
(but for other algorithms they can be different), 
and our goal is to achieve $\tilde{O}(k)$ span for a perfect parallelization.

Note that the ``perfect parallelization'' of a sequential algorithm does not directly suggest optimality in work or span bounds for the same \emph{problem}. 
One can possibly achieve better bounds by redesigning the recurrence and/or sequential algorithm with fewer edges or a shallower depth.
%As mentioned, some shortcut methods~\cite{} allows to significantly redesign the DAG and achieve better worst-case span than $\dep^*(G_{\Gamma^+})$.
%We note that this is out of the scope of this paper. 
Instead of finding or redesigning a different DAG to obtain new optimizations, 
our focus is to provide parallelization of existing sequential algorithms with optimizations. 

In the following, we first present our new algorithmic framework: the \ouralgo{},
which provides a correct, although not necessarily efficient parallelization for general DP algorithms. 
%In later sections, we provide additional techniques to achieve efficient work using the \ouralgo{}. 
On top of it, for each specific problem, 
we will show how to achieve low work, 
which, as discussed in \cref{sec:post-office,sec:other}, can be highly non-trivial.

\subsection{Our Framework: the \ouralgo{}}\label{sec:framework}

Our idea is based on the \emph{phase-parallel} framework~\cite{shen2022many} (see below)
adapted to DP algorithms. 
%, and is a phase-parallel solution for DP algorithms. %or any algorithms with dependencies defined by DAGs.
%The \ouralgo{} gives a general guidance on how to parallelize DP algorithms.
%\subsection{The Algorithm Description}
The phase-parallel framework by Shen et al. aims to identify (as many as possible) operations that do not depend on each other, and process them in parallel. 
Directly applying this framework to DP algorithms will give the following algorithm outline:

\vspace{.3em}

\noindent\textbf{Phase-Parallel Framework for DP algorithms}
\begin{mdframed}[style=mystyle]
\begin{enumerate}[nosep,label=\arabic*., leftmargin=*]
  \item[] \label{step:pp:while}While there exist any tentative states:
  \begin{itemize}[nosep, leftmargin=2em]%,label=\arabic*.,
    \item Find the set of ready states as $\ff$
    \item Process all states in $\ff$ in parallel and mark all of them as finalized 
    %\item $i\gets i+1$
  \end{itemize}
\end{enumerate}
\end{mdframed}
\vspace{.3em}

We call each iteration of the while-loop a \emp{round}. 
We call the set of states being processed in round $i$ the \emp{frontier} of round $i$, denoted as $\ff_i$. 
While the phase-parallel framework gives a high-level approach in achieving parallelism, it does not indicate how to do so (i.e., how to identify the ready states in each round).
%The key part to use this framework is then to identify the ready states properly. 
%Clearly, all finalized states must form a connected upper part of the DAG. 

We now introduce our \ouralgo{}, which uses a novel approach to identify the frontier $\ff_{i}$ in each round, particularly for a DP computation $G_\Gamma$. 
%present our algorithm in \cref{alg:ours}. 
%In a high-level, it still follows the phase-parallel framework, but employs a novel approach to identify the frontiers. 
The \ouralgo{} identifies the unready tentative states and put \emp{\sentinel{s}} on them; then it uses all \sentinel{s}
to outline a \emp{cordon} to mark the boundary of the frontier. We summarize the algorithm in the following steps.
Note that every step can be processed in parallel. 

\begin{enumerate}[label=\textbf{Step \arabic*}, wide]
  \vspace{-0.05cm}
  \item \label{step:ours:init} Mark all states as tentative and initialize them by the boundary condition.
  \vspace{-0.05cm}
  \item \label{step:ours:find-ready} 
  If a tentative state $j$ can successfully relax another tentative state $i$ (i.e., update $D[i]$ to a better value), 
  put a \sentinel{} on state $i$. Such a \sentinel{} means that all the descendants (inclusive) of state $i$ are unready. 
  We say this state and all its descendants are \emp{blocked} by the \sentinel{} in this case. 
  %If a tentative state $i$ can be successfully relaxed (i.e., updated to a better value) by another tentative state $j$, put a \sentinel{} on state $i$. Such a \sentinel{} means that all the descendants (inclusive) of this state are unready. We say this state and all its descendants are \emp{blocked} by this \sentinel{} in this case. 
      Therefore, a state is \textbf{ready} if there is no \sentinel{} on any of its ancestors (inclusive). 
  \vspace{-0.05cm}
  \item \label{step:ours:relax} For each ready state, use its DP value to relax the tentative DP values of its descendants. Usually, we need to do so \emph{implicitly} to achieve efficient work. We discuss more details later.
  \vspace{-0.05cm}
  \item \label{step:ours:mark-finalized} Mark all ready states as finalized. Clear all the sentinels.
  \vspace{-0.05cm}
  \item \label{step:ours:repeat} If there still exist tentative states, go to \ref{step:ours:find-ready} and repeat. 
\end{enumerate}

We will first prove that the algorithm is correct, i.e., it computes correct DP values for all states. 
Later we will show some motivating examples to help understand this algorithm in \cref{sec:lis}. 
%, and then applied it to solve DP algorithms with decision monotonicity. 

\begin{figure*}[t]
  \centering
  \includegraphics[width=2\columnwidth]{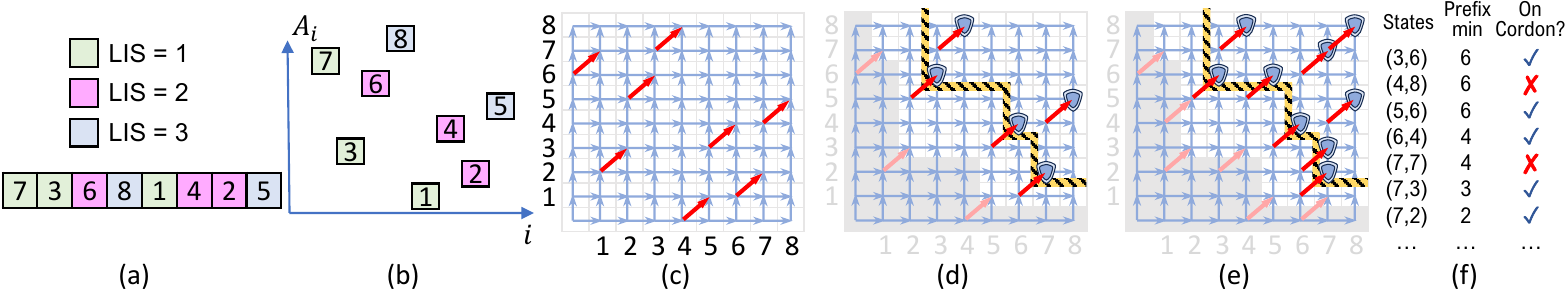}
  \caption{\small {Illustrations for the LIS/LCS problem and the \ouralgo{}.}
  Subfigure (a): An input sequence for LIS with the DP value of each element.
  Subfigure (b): A geometric view of the input sequence on a 2D plane with each element represented as $(i, A_i)$.
  Subfigure (c): The corresponding LCS on this input---the answer is the longest path from $(0,0)$ to $(8,8)$ using the maximum number of red edges.
  Subfigure (d): The process to compute the second cordon.  The ready states are marked in the shaded gray region.  The cordon is decided by the three cells with LIS=2 in the original input.
  Subfigure (e): A general LCS case where every diagonal can be a red edge.
  Subfigure (f): An example execution of our LCS algorithm. Here we only show unready states for simplicity.
  % \yihan{remove ``rank'' in the figure.}
  \label{fig:lis}}

\end{figure*} 

\begin{theorem}
  The \ouralgo{} is correct.
\end{theorem}
\begin{proof}
  This is equivalent to show that, when we find a ready state and later mark it finalized, its DP value must be finalized. 
  
  We will show this inductively. At the beginning, the ready states are those with zero in-degree in the DAG, and their true DP values are specified by the boundary cases. 
  Clearly, their DP values cannot be relaxed by other states and they will be identified as ready in our algorithm. 
  Their DP values are also computed by the boundary in \ref{step:ours:init}. Therefore, our algorithm is correct in the first round. 
  
  Assume the algorithm is correct up to round $r$. We will show that round $r+1$ also correctly finds the ready states and their DP values. 
  Assume to the contrary that a state $i$ identified in \ref{step:ours:find-ready} does not have its true DP value yet. 
  %Therefore, the best decision of $i$ have not relaxed state $i$ to a better value. 
  This means that the best decision of \stt{$i$} has not relaxed \stt{$i$} to the finalized value. 
  Let $j=\best[i]$. 
  Note that \stt{$j$} cannot be finalized: if so, before $j$ is marked as finalized in \ref{step:ours:mark-finalized}, in \ref{step:ours:relax} it must have relaxed \stt{$i$}.
  \hide{We discuss two cases. 
  First, \stt{$j$} is finalized. This is impossible because before $j$ is marked as finalized in \ref{step:ours:mark-finalized}, in \ref{step:ours:relax} it must have relaxed \stt{i}, which leads to a contradiction.
  }
  
  %The second case is when \stt{$j$} is tentative. 
  Therefore, \stt{$j$} is tentative. 
  In this case, let $j_0=i$, and state $j_x$ be the best decision of state $j_{x-1}$ for $x\ge 1$, i.e., we chase the chain of best decisions and get a list of states $i=j_0, j_1=\best[j_0]$, etc.
  Let us find the first $x$ such that $D[j_x]$ is not the true DP value but $D[j_{x+1}]$ is the true DP value. 
  We then prove that state $j_{x+1}$ must be a tentative state. If we consider $j_y$ as the first finalized state on this chain, 
  then $j_{y-1}$ is a tentative state, and also must already have its true DP value (because $j_y$ has relaxed it in \ref{step:ours:relax}). 
  Note that since state $j_{y-1}$ is a tentative state, all states between $j_0$ and $j_{y-1}$ must be tentative. 
  Since $j_{y-1}$ has its true DP value, and $i=j_0$ does not have its true DP value, 
  the first switch point $j_x$ must be between $j_0$ and $j_{y-1}$, and must also be a tentative state.
  
  Therefore, state $j_{x+1}$ is a tentative state that can relax state $j_{x}$, so it will put a \sentinel{} on state $j_x$. As a descendant of $j_x$, 
  state $i$ must be blocked by $j_x$, and cannot be identified as a ready state. This leads to a contradiction.   
  Therefore, if a state $i$ is identified as ready in our algorithm, its DP value must have been finalized.   
\end{proof}

%The \ouralgo{} can correctly compute the DP values

%\input{alg/cordon.tex}

The \ouralgo{} tells which states/vertices should be in the frontier in each round, but the algorithm does not show how to do so \emph{efficiently}. %, which is algorithm-specific. 
Especially in \ref{step:ours:relax},
it is almost infeasible to use the finalized DP values to explicitly update all other states.
In this case, we have to develop new parallel data structures to facilitate this step. 
Next, we first use LIS and LCS as two examples to illustrate our framework.
We then apply it to more involved cases that use decision monotonicity in \cref{sec:post-office,sec:other}.

\section{Motivating Examples on LIS/LCS}\label{sec:lis}

%Most of the new algorithms we proposed in this paper are fairly complicated---the sequential DP algorithms with decision monotonicity are already involved in most cases.
%Hence, we feel necessary to provide a running example to help audience understand the core concepts introduced above we introduce our novel parallel DP algorithms in \cref{sec:post-office}.
To help the audience understand the more complicated algorithms using DM in the following sections,
we first provide two simple examples on \ouralgo{}, especially on how to compute the cordon efficiently.

\myparagraph{Longest Increasing Subsequence (LIS).} We first use LIS as an example.
Given an input sequence $A_i$, LIS computes the maximum of $D[i]$ such that:
\begin{equation}\label{eqn:lis}
D[i]=\max\{1, \max_{j<i, A_j<A_i} D[j]+1\}
\end{equation}
We use $n$ as the input size and $k$ as the output LIS length. 
Directly computing this recurrence takes $O(n^2)$ work.
Sequentially, we can process all states one by one, and compute $\max_{j<i, A_j<A_i} D[j]$ using a binary search structure, giving $O(n\log k)$ total cost~\cite{Knuth73vol3}. 
The binary search precisely finds the best decision of each state, and only $n$ transitions are processed. 
%In parallel Gu et al. proposed a work-efficient algorithm with $\tilde{O}(k)$ span, where $k$
Many parallel LIS solutions have also been proposed~\cite{shen2022many,gu2023parallel,cao2023nearly,krusche2009parallel,krusche2010new}.
We will show that solving LIS using our \ouralgo{} framework will essentially
give an existing algorithm~\cite{gu2023parallel} 
and is a perfect parallelization of the sequential $O(n\log k)$ algorithm. 

Based on \ouralgo{}, the boundary case is to set all tentative DP values as 1. 
Then we will attempt to use the current tentative DP values for relaxation. 
In this case, for a state $i$, as long as there is any other state $j<i$ 
with $A_j<A_i$, $D[i]$ can be relaxed to a better value 2. 
Therefore, all ready states are those input objects that are \emph{prefix-min}
elements in the sequence, i.e., $A_i$ is the smallest value among all $A_{1..i}$.
We can set these states as ready, update all other states and repeat.
Note that since all unfinalized states have the same tentative DP value of 2,
we \emph{do not need to explicitly update} the values in $D[i]$, but
can just maintain a global variable as the current tentative DP value.
By the same idea, the ready states in the next round would be the prefix-min elements
in the input after removing the finalized states. 
The same observation (repeatedly finding prefix-min elements) is exactly the core idea of the algorithm in~\cite{gu2023parallel}.
In their algorithm, they further use a tournament tree to identify prefix-min elements efficiently and achieve efficient $O(n\log k)$ work, 
and $O(k\log n)$ span. 
Note that $k$ is exactly the perfect depth of the DP DAG, which is the longest dependency between best decisions. 

\begin{theorem}\label{thm:lis}
  Combining with a tournament tree, \ouralgo{} leads to a perfect parallelization of sequential $O(n\log k)$ LIS algorithm in~\cite{gu2023parallel},
  where $n$ is the input size, and $k$ is the LIS length.
\end{theorem}

\myparagraph{Longest Common Subsequence (LCS).}
LIS has a close relationship with other important problems such as LCS. 
Here we will revisit our cordon-based LIS algorithm from the view of LCS, which also leads to a new parallel LCS algorithm. 
%More interestingly, in \cref{xx} we will show how to revisit the LIS algorithm as an LCS problem,
%and show how this leads to an efficient parallel LCS algorithm based on the \ouralgo{}. 
Given two sequences $A[1..n]$ and $B[1..m]$ ($m\le n$), LCS aims to find a common subsequence $C$ of $A$ and $B$ such that $C$ has the longest length among all
common subsequences. 
An LIS problem can be reduced to an LCS problem by first relabeling all input elements by $1..n$ based on their total order,
and then finding the LCS between this new sequence and a sequence $B=\langle1,2,\dots, n\rangle$. See \cref{fig:lis}(a)--(c) for an illustration.
LCS has been extensively studied both sequentially \cite{hunt1977fast,hirschberg1977algorithms,apostolico1987longest,eppstein1992sparse,Knuth73vol3}
and in parallel \cite{lu1994parallel,tchendji2020efficient,babu1997parallel,xu2005fast,apostolico1990efficient,chen2006fast,ding2023efficient}.
The standard DP solution defines each state $D[i,j]$ as the LCS for $A[1..i]$ and $B[1..j]$, and uses the following recurrence:
\begin{align}\small\label{eqn:lcs}
  D[i,j]=\begin{cases}
           0, & \mbox{if } i=0 \mbox{ or } j=0 \\
           D[i-1,j-1]+1, & \mbox{if } A[i]=B[j] \\
           \max\{D[i-1,j],D[i,j-1]\}, & \mbox{otherwise}.
         \end{cases}
\end{align}

These transitions correspond to horizontal, vertical and diagonal edges on a grid (see an example in \cref{fig:lis}(c)). 
A known sequential optimization (i.e., sparsification)~\cite{hunt1977fast,hirschberg1977algorithms,apostolico1987longest,eppstein1992sparse} to this recurrence is to observe that only the edges correspond to the diagonal edges with $A[i]=B[j]$ are useful. 
The computation is equivalent to finding the longest path from the bottom-left to top-right corresponds to these effective (red) edges,
and all other edges and states can be skipped. This can lead to a sequential algorithm with $O(\totarrows{}\log n)$ cost,
where $\totarrows{}$ is the number of pairs $(i,j)$ such that $A[i]=B[j]$. For LIS, there are exactly $\totarrows{}=n$ such effective edges.

We will show how \ouralgo{} can be used to parallelize this optimization. 
Starting with the boundary where $D[i,j]=0$, we observe that the DP value $D[i,j]$ of any state with $A[i]=B[j]$
can be updated to a \emph{better} value. Therefore, we will put a \sentinel{} at 
each of such states to indicate that they should be updated. All such \sentinel{s} will block the top-right part of the 
grid. In this way, the blocked region is clearly marked by a \emp{staircase} region, as shown in \cref{fig:lis}(d). 
Therefore, the entire region within the first cordon has the DP value 0.
By repeatedly doing this, we will effectively find that the region between the cordons of round $i+1$ and $i$ are those
states with DP value (LCS length) $i$. The algorithm finishes in $k$ rounds where $k$ is the LCS length.

The problem boils down to efficiently identifying the cordon (i.e., the staircase) in each round.
Note that in LIS there is at most one effective edge in each column (see \cref{fig:lis}(d)),
while in LCS there can be multiple effective edges in each column (see \cref{fig:lis}(e)).
We will show an interesting modification to the original LIS algorithm that can handle this more complicated setting.
Here we sort all edges by column index as the primary key (from the smallest to largest) and row index as the secondary key (but from largest to smallest).
An example is illustrated in \cref{fig:lis}(f).
Then, we will still use a tournament tree to maintain this list, and apply prefix-min on the row indexes.
It is easy to see that a state/edge is on the cordon if its row index is smaller than or equal to the prefix-min.
A tournament tree can identify, mark, and remove these states in $O(l\log (L/l))$ work and $O(\log n)$ span~\cite{gu2023parallel}, where $l$ is the number of diagonal edges on the cordon.
We thus have the following theorem.

\begin{theorem}\label{thm:lcs}
  Combining with a tournament tree, \ouralgo{} leads to a perfect parallelization ($O(\totarrows{}\log n)$ work and $O(k\log n)$ span) of sequential LCS algorithm in~\cite{apostolico1987longest},
  where $n$ and $m<n$ are the input sequence sizes, $\totarrows$ is the number of pairs $(i,j)$ such that $A[i]=B[j]$, and $k$ is the LCS length. 
\end{theorem}

Since $L=O(n^2)$, the $O(\log L)$ terms in the cost of the tournament trees is stated as $O(\log n)$ in the theorem.

Interestingly, to the best of our knowledge, this is the first parallel LCS algorithm with $o(mn)$ work and $o(\min(n,m))$ span for {sparse} LCS problem (i.e., $\totarrows=o(mn)$ and $k=o(\min(n,m))$). 
Meanwhile, this algorithm is quite simple---we provide our implementation in~\cite{dpdpcode} and experimentally study it in \cref{sec:exp}.
Another interesting finding is that our algorithm implies how to map LCS to LIS (we only know the other direction).
Given two input strings $A$ and $B$, if we sort all $(i,j)$ pairs for $A[i]=B[j]$ by increasing $i$ (primary key) and decreasing $j$ (secondary key), then LCS is equivalent to the LIS on the secondary keys (the $j$(s)) of this sorted list.

We will show more sophisticated parallelization of DP algorithms in the next sections.
Our LCS algorithm will be a subroutine in the more involved parallel GAP algorithm introduced in \cref{sec:gap}.

\section{Parallel Generalized LWS}\label{sec:post-office}

We now discuss the convex/concave \emp{generalized least weight subsequence (GLWS)} problem, 
which is one of the most classic cases of decision monotonicity (DM).
%One classic example of decision monotonicity (DM) is the convex/concave \emp{generalized least weight subsequence (GLWS)} problem. 
Given a cost function $w(j,i)$ for integers $0\le j<i\le n$ and $D[0]$, 
the GLWS problem computes:
\begin{equation}\label{dp-equation}
  \displaystyle D[i]=\min_{0\le j<i}\{E[j]+w(j,i)\}
\end{equation}
for $1\le i \le n$, where $E[j]=f(D[j],j)$ can be computed in constant time. 
The original least weight subsequence (LWS) problem~\cite{hirschberg1987least} is
a special case when $E[j]=D[j]$.  
Here we use the general case $E[j]=f(D[j],j)$ that
has the same sequential work bound~\cite{eppstein1990sequence,galil1989linear,klawe1989simple,galil1992dynamic,eppstein1988speeding,eppstein1992sparse2}, 
because the generalized version is needed in many applications (see examples in \cref{sec:other}).
The GLWS problem is also referred to as 1D/1D DP by \citet{galil1992dynamic}.
%The GAP problem in \cref{sec:gap} also needs this generalized version.
% Given a sequence of objects, and a weight function $w(i,j)$ that defines the \emph{weight} of the subsequence from the $i$-th to the $j$-th objects,
% the LWS problem aims at split the objects into subsequences and minimize the total weight.
The GLWS problem is highly relevant to other important problems (e.g., line breaking~\cite{knuth1981breaking},
optimal alphabetic trees \cite{larmore1993parallel}, and a number of computational geometry problems \cite{aggarwal1990applications}).
%The essence of GLWS is to cluster a list of 1D objects based on the spatial proximity, which is a general approach for \emph{1D clustering}.
The essence of GLWS is to cluster a list of 1D objects based on spatial proximity and minimize the total weighted sum of all clusters.
As an intuitive example, consider selecting a subset of villages on a road (with their coordinates known) to build post offices to minimize the total cost, 
where 
$w(j,i)$ is the cost of using one post office to serve villages $j+1$ to $i$. 
%building a post office to serve villages $j+1$ to $i$ incurs a fixed constant cost $c$, and some cost 
%related to the range of villages, denoted as $w(j,i)$. 
This gives a GLWS problem with $D[i]$ as the lowest cost to serve the first $i$ villages and $E[i]=D[i]$. 
The DP recurrence enumerates all possible decisions $j$ such that the last post office serves the villages $j+1$ to $i$,
and takes the minimum cost among all possible decisions $j$. 
Many cost functions $w$ used in practice (e.g., a fixed cost plus a linear or quadratic cost to the service range or sum of distances from villages to the post office) are \emph{convex}, 
which implies decision monotonicity (DM) --- for two states $i$ and $j>i$, their best decisions $i^*$ and $j^*$ satisfy $j^*\ge i^*$.
Symmetrically we can show that for \emph{concave} cost functions~$w$, either $j^*\ge i$ or $j^*\le i^*$ holds, although concave cost functions are less common in the real-world applications of GLWS.
\hide{
Many practical cost functions $w$, e.g., $w(j,i)$ is the total distance from the post office to all villages in this range,
are \emph{convex} (defined below) and lead to DM. 
%Namely, the $\best[\cdot]$ array is non-decreasing. 
This means that for two states $i$ and $j>i$ where their best decisions are $i^*$ and $j^*$, then we must have $j^*\ge i^*$. 
%One can also define a concave cost function, where $\best[\cdot]$ is non-increasing. 
One can also define a \emph{concave} cost function, where we must have either $j^*\ge i$ or $j^*\le i^*$. }

Given its high relevance in practice, 
%there has been existing work tried to study parallel solutions for convex GLWS.
convex GLWS has been studied in parallel. 
\citet{apostolico1990efficient} showed an algorithm with $O(n^2\log n)$ work and $O(\log^2n)$ span. 
%Since the total length for all intervals in one recursive iteration is at most $n$, in the worst case when one interval has size $O(n)$, the bounds are dominated by this single subproblem.
%Hence, the entire algorithm requires $O(n^2\log^2n)$ work and $O(\log^3n)$ span.
Later, \citet{larmore1995constructing} showed an improved algorithm with $O(n^{1.5}\log n)$ work and $O(\sqrt{n}\log n)$ span.
Despite the interesting algorithmic insights in these algorithms, the polynomial overhead in work limits their potential to outperform the classic sequential solutions with $\tilde{O}(n)$ work~\cite{eppstein1990sequence,galil1989linear,klawe1989simple,galil1992dynamic,eppstein1988speeding,eppstein1992sparse2}.
For the concave case, some work \cite{chan1990finding,bradford1998efficient} achieved near work-efficiency and polylog span on the original LWS,
but the ideas cannot be applied to generalized LWS.

In this section, we show how to use the \ouralgo{} to parallelize a well-known sequential GLWS algorithm with ${O}(n\log n)$ work, which works for both convex and concave DM. 
Although efficiently applying \ouralgo{} here requires many sophisticated algorithmic techniques, 
our parallel algorithm (\cref{alg:dp-glws}) remains practical and it indeed outperforms the sequential algorithm in a wide parameter range (see \cref{sec:exp} for details).
It is also the key building block for many other algorithms shown later in \cref{sec:other}.

We start with preliminaries and the classic sequential algorithm, then discuss how to use \ouralgo{} to parallelize it.
We will use the convex case when describing the algorithm since it is used more often in practice, and discuss the concave case in \cref{sec:post-office-concave}.
%As mentioned, convex weight functions are more frequently encountered in real-life situations.
%Therefore, we will mainly focus on the convex case. 
%The adaptations of our algorithm to handle the concave case will be discussed in
%in \cref{sec:post-office-concave}.

\subsection{Preliminaries}\label{sec:post-office:prelim}

\myparagraph{Convexity, Concavity and Decision Monotonicity.}
The convexity of the cost function $w$ is defined by the Monge condition \cite{monge1781memoire}.
We say that $w$ satisfies the \defn{convex} Monge condition (also known as quadrangle inequality \cite{yao1980efficient}) if for all $a<b<c<d$,
\begin{equation}
  w(a,c)+w(b,d) \le w(b,c)+w(a,d).
\end{equation}
We say that $w$ satisfies the \defn{concave} Monge condition (also known as inverse quadrangle inequality) if for all $a<b<c<d$,
\begin{equation}
  w(a,c)+w(b,d) \ge w(b,c)+w(a,d).
\end{equation}
%The convex Monge condition is also known as the quadrangle inequality \cite{yao1980efficient},
%and the concave Monage condition is also known as the inverse quadrangle inequality.

Consider two states $i$ and $j>i$ with best decisions $i^*$ and $j^*$. 
A convex weight function leads to DM such that $j^*\ge i^*$. 
A concave weight function leads to DM such that either $j^*\ge i$ or $j^*\le i^*$.

\iffullversion{
Another condition closely related to the Monge condition is the total monotonicity \cite{aggarwal1987geometric}.
We say a $n\times m$ matrix $A$ is \defn{convex totally monotone} if for $a<b$ and $c<d$,
$$
A(a,c)\ge A(a,d) \Rightarrow A(b,c)\ge A(b,d).
$$
We say a $n\times m$ matrix $A$ is \defn{concave totally monotone} if for $a<b$ and $c<d$,
$$
A(a,c)\le A(a,d) \Rightarrow A(b,c)\le A(b,d).
$$
Let $c_i$ be the column index such that $A[i,c_i]$ is the minimum value in row $i$.
The convex total monotonicity of $A$ implies that $r_1\le r_2\le...\le r_n$,
while in the concave case $r_1\ge r_2\ge...\ge r_n$ \cite{galil1992dynamic}.
If $A$ is convex totally monotone, any submatrix $B$ of $A$ is convex totally monotone.
The convex/concave Monge condition is the sufficient but not necessary condition for convex/concave total monotonicity.

In GLWS, we call $f_{i,j}=E[j]+w(j,i)$ a transition from $j$ to $i$.
The convex/concave decision monotonicity is equivalent to the convex/concave total monotonicity of $f_{i,j}$.
Note that if $w$ satisfies the convex/concave Monge condition, so does $f_{i,j}$.
But the convex/concave total monotonicity of $w$ does not guarantee the convex total monotonicity of $f_{i,j}$.
All the theorems in this paper only need the convex/concave total monotonicity of $f_{i,j}$.
}

\ifconference{
A more general condition to enable DM is \emph{convex/concave total monotonicity}~\cite{aggarwal1987geometric}.
Our algorithm works as long as the transition function is convex/concave totally monotone. 
For page limit, we provide the full background of related discussions in~\cite{fullversion}.
}

\myparagraph{The Sequential Algorithm.}
The best (sequential) work bound for convex GLWS is $O(n)$~\cite{larmore1991line,klawe1989simple,galil1992dynamic}, and $O(n\alpha(n))$ for the concave case~\cite{klawe1990almost}. %
%There exists an $O(n)$ algorithm for the convex GLWS
%and an $O(n\alpha(n))$ algorithm for the concave version\cite{galil1992dynamic}.
However, both of them are mainly of theoretical interest since they are complicated and have large constants in both work and space usage. 
%and thus are less likely to be practical. 
%but due to large time and space constant, 
%parallelizing them is off theoratical interest.
%Thus, due to the practical consideration, we parallelize a simple algorithm with $O(n\log n)$ work, 
We parallelize a simpler and more practical algorithm with $O(n\log n)$ work~\cite{galil1992dynamic}.
%Here we describe the $O(n\log n)$ algorithm for the convex case given in \cite{galil1992dynamic},
%which is simpler and widely used in implementations. 
This algorithm computes $D[1..n]$ in order.
%Instead of finding the best $j$ for each state $i$,
It implicitly maintains the best decision array $\best[1..n]$.
When the algorithm finishes computing $D[i]$, the algorithm updates $\best[(i+1)..n]$ using $D[i]$,
then $\best[j]$ ($j>i$) will be the best decision of state $j$ among states $0$ to $i$. 
% that stores 
%As such, and $D[i]$ can be directly computed from $\best[i]$.
%Then the array $\best[i+1..n]$ will be updated based on the new value of $D[i]$.
%If $i$ can successfully relax another state we should update its best decision to $i$.

However, maintaining and updating this array of size $n$ for $n$ iterations require quadratic work.
%Observe that in the convex case, $\best[i+1..n]$ are always non-decreasing, and in the concave case $\best[i+1..n]$ are always non-increasing. 
Observe that after computing $D[i]$, $\best[(i+1)..n]$ must be non-decreasing in the convex case, 
and must be non-increasing in the concave case~\cite{galil1992dynamic}. 
Hence, the algorithm maintains a ``compressed'' version of $\best[(i+1)..n]$ by a list of triples $([l,r],j)$,
which indicates that all states between $l$ and $r$ have best decision $j$, i.e., $\forall i'\in [l,r]$, $\best[i']=j$. 
The list is maintained by a \emph{monotonic queue}, %with $O(1)$ amortized cost for each insertion and deletion. 
% Thus, we need to redesign it in the parallel setting. 
which is a classic data structure based on double-ended queue, and is inherently sequential. 
In the $i$-th iteration, we can directly find the decision of state $i$ from the queue. 
%Then, when updating the monotonic queue, a binary search with $O(\log n)$ cost is need to compute the region $([l_i,r_i],i)$, which will be merged into the monotonic queue.
After obtaining $D[i]$, the monotonic queue can be updated in $O(\log n)$ amortized cost to consider $i$ as a decision for all later states.
%In total, this algorithm incurs $n$ binary searches, and $n$ insertions and deletions to the monotonic queue, with $O(n\log n)$ work in total.
In total, this algorithm costs $O(n\log n)$ work.
Here we refer to the audience to the original paper for details of this algorithm. 
We will call this algorithm $\gammalws$. 
Making use of DM, $\gammalws$ only processes transitions between each state $i$ and its best decision. 
%The goal is to find a cutting point $p$ such that the best decisions after/before $p$ are updated to $i$ in the convex/concave case, respectively.
The DAG $G_{\gammalws}$ for this algorithm includes normal edges $\edge{j}{i}$ for all $j<i$,
and exactly $n$ effective edges $\edge{\best[i]}{i}$ for all states $i$. 

\hide{
To efficiently batch-update the best decision array,
the sequential algorithm maintains double-ended queue storing triples $([l,r],j)$,
which means that $\best[l..r]=j$.
The triples in the queue are sorted in ascending order of $[l,r]$
so that $\best[i]$ can be read directly from the triple in the front of the queue.
After we compute $D[i]$,
in the convex case we can start from the back of the queue (or from the front in the concave case)
and check whether $i$ is better than the current best decision.
For a interval we only check its two endpoints, and if $i$ can successfully relax both of them we can pop out this interval and proceed to the next one.
Otherwise, we do a binary search inside this interval to find the cutting point.
}

\hide{
Duering the process we add $O(n)$ intervals to the queue,
so the push/pop of the invervals taks $O(1)$ time amortized for each $i$.
Plus the cost of binary search, the total work is $O(n\log n)$.
}

Due to simplicity, this algorithm is usually the choice of implementation in practice. 
We will show a parallel version of this algorithm using \ouralgo{}. 

\hide{
\myparagraph{Existing Parallel Algorithms}
For the LWS problem where $f(D[j])=D[j]$ and $w(j,i)$ is concave,
there exists a near work-efficient algorithm with polylogarithmic span~\cite{chan1990finding}.
We are unaware of existing (nearly) work-efficient and well-parallelized algorithm for the generalized LWS, 
whether for convex or concave cases.}

\subsection{Parallel Convex GLWS}

\begin{algorithm}[t]
\small
\caption{Parallel Convex GLWS \label{alg:dp-glws}}
\SetKwProg{myfunc}{Function}{}{}
\SetKwFor{parForEach}{ParallelForEach}{do}{endfor}
\SetKwFor{mystruct}{Struct}{}{}
\SetKwFor{pardo}{In Parallel:}{}{}
\SetKwInput{Input}{Input}
\Input{
problem size $n,D[0]$, cost function $w(\cdot, \cdot)$
}
\SetKwInput{Output}{Output}
\Output{
$D[1..n]$: the DP table
}
\SetKwInput{Maintains}{Maintains}
\Maintains{
%$\now$: the largest position where $D[1..\now]$ are finalized \\
$B$: an sorted array storing triples of $([l,r],j)$, meaning that for all $l\le i\le r$, the current best decision $\best[i]=j$ \\
}

$\now \gets 0$  \\
\While(){$\now < n$}{
    $\nxt \gets$ \funcfont{FindCordon}($\now$) \\
    %Compute $D[i]$ for $i\in[\now+1..\nxt-1]$ \label{line:1d:compute}\\
    $\funcfont{UpdateBest}(\now,\nxt)$ \label{line:updateB}\\
    $\now \gets \nxt-1$
}
\Return $D[1..n]$

\DontPrintSemicolon

\myfunc{\upshape\funcfont{FindCordon}($\now$)}{
    %$t \gets 1$  \\
    $\cordon \gets n+1$ \\
    %\While(){true}{
    \For{$t\gets 1$ \emph{\textbf{to}} $\log n$\label{line:1d:forloop}}{
        $l \gets \now+2^{t-1}$ \\
        $r \gets \min(n,\now+2^t-1)$ \\
        \parForEach{$j\in[l,r]$}{
        %Read the current $\best[j]$ from $B$. \\
        %Binary search $s_j$ on $T$, where $s_j$ is the minimum position that $j$ can successfully relax $s_j$ (the sentinel put by $j$).
        Let $\best[j]$ be the current best decision of $j$ recorded by $B$\\
        $D[j]\gets E[\best[j]]+w(\best[j],j)$\tcp*[f]{relax $j$}\label{line:1d:computeD}\\
        Binary search in $B$ and find $s_j=\min \{i: E[j]+w(j,i) < E[\best[i]]+w(\best[i],i)\}$, 
        i.e., $i$ is the first state that can be successfully relaxed by $j$ \label{line:1d:try-relax}
        }
        $\cordon\gets\min(\cordon,\min_{l\le j\le r}s_j)$ \\
        \lIf(){$\cordon\le r+1$}{break}
        %$t \gets t+1$
    }
    \Return $\cordon$
}

\myfunc{\upshape\funcfont{UpdateBest}($\now,\nxt$)}{
    Tree $T\gets \funcfont{FindIntervals}(\now+1,\nxt-1,\nxt,n)$ \\
    Flatten $T$ into array $B$\\
    Merge adjancent intervals with the same best decision in $B$ \\
}

% \tcp{Use $D[\jl..\jr]$ to update $\best[\il..\ir]$}
%\yihan{We may want to change the parameters to be something more clear, instead of $i_l,i_r,j_l,j_r$}\\
\myfunc{\upshape\funcfont{FindIntervals}($\jl,\jr,\il,\ir$)\tcp*[f]{use $D[\jl..\jr]$ to update $\best[\il..\ir]$}}{
    \lIf(){$\il>\ir$}{\Return \emph{null}}
    \lIf{$\jl=\jr$}{
      \Return a leaf node$([\il,\ir],\jl)$ 
    }
    
        $\im \gets (\il+\ir)/2$ \\
        $\jm \gets \argmin_{\jl\le j \le \jr}(E[j]+w(j,\im))$\label{line:1d:jm} \\
        $x\gets$ node$([\im,\im],\jm)$ \\
        \pardo{}{
        $T_l\gets$\funcfont{FindIntervals}($\jl,\jm,\il,\im-1$) \\
        $T_r\gets$\funcfont{FindIntervals}($\jm,\jr,\im+1,\ir$)
        }
        \Return node $x$ with left child as $T_l$ and right child as $T_r$
    
}

\end{algorithm}

We first give the parallel algorithm of convex GLWS. 
We will use the ``post-office'' problem mentioned above as a running example to explain the concepts, 
but our algorithm works for general cases. 

Following the idea of the phase-parallel algorithm, 
with the current finalized states, the goal is to find all ready states as the frontier, 
where their true DP values can be computed from the finalized ones. 
We will use our \ouralgo{} to find the frontier in each round. 
Na\"ively, the recurrence suggests that a state depends on all states before it.
%However, note that with the current finalized states, 
%any state $i$ with $\best[i]$ finalized are effectively ready. 
However, note that a state is essentially ready as long as its best decision has been finalized. 
For the convex case, we will use the fact as shown below. 

\begin{compactfact}
  In convex GLWS, let $S=\{j: j>i\wedge\best[j]\le i\}$, which is the set of states with best decisions no later than 
  state $i$; then $S$ is a consecutive range of states starting from $i+1$. 
\end{compactfact}

This is a known fact in the sequential setting (can be proved by induction). It suggests that
the frontier of each round in the phase-parallel algorithm is a consecutive range of states. Based on this idea, 
we will maintain $\now$ as the last finalized state in each round. 
Then in the next round, ideally, the algorithm should find the cordon at $\nxt$, 
where all states $[\now+1, \nxt-1]$ are ready and can compute their true DP value from (i.e., have their best decisions at) states no later than $\now$. 
We show an example of the post-office problem to illustrate the phase-parallel framework in \cref{fig:1d-parallel}.
%The ready states are those with their best decisions finalized.
Based on the discussions above, ideally, in round $i$, the ready states in the frontier
are those where the best solution contains $i$ post offices. 
This is because their best decision must have $i-1$ post offices, and must be finalized in the previous round. 

This high-level idea is presented in the main function of \cref{alg:dp-glws}. 
Starting from $\now=0$, given the current finalized states $[0,\now]$, 
we will find all ready states $[\now+1,\nxt-1]$ using the \ouralgo{}, 
which essentially will find the cordon at $\nxt$. 
We explain this part with more details in \cref{sec:cordon}. 
% Then we explicitly update the DP values for these ready states (\cref{line:xx}). 
%use the new DP values to update the current best decisions for future states. 
%Finally, with the new DP values, we will update 
Similar to the sequential algorithm,
we also maintain a data structure $B$ to store all triples $([l,r],j)$ in order,
which indicates that all states between $l$ and $r$ have best decisions at $j$. 
This data structure is essential to guarantee the efficiency of finding the next frontier, 
and also has to be updated after each round with the new DP values (\cref{line:updateB}).

\subsubsection{Finding the Cordon}\label{sec:cordon}

To find the ready states in each round, we use the \ouralgo{}.
Namely, with all states up to $\now$ finalized, 
we can attempt to use the tentative states after $\now$ to update other tentative DP values. 
Once we find any \stt{$j$} that can update \stt{$i$},
we put a \sentinel{} at \stt{$i$}. Among all \sentinel{}s, the smallest (leftmost) one will give the final position of the cordon. 

However, note that we cannot afford exhaustive checking for all pairs of states $(j,i)$. 
First of all, checking all possible $j>\now$ may incur large overhead in work,
since most of the later states are unready anyway. 
Ideally, the algorithm should check up to exactly the position of $\nxt$, but this would be a chicken-and-egg problem. 
To handle this, our idea is to use \emph{prefix-doubling}, a common idea in parallel algorithm design (e.g.,~\cite{blelloch2016parallelism,gu2023parallel,gu2022parallel,shen2022many,gbbs2021}) to achieve work-efficiency and high parallelism.
Here, prefix-doubling is used in function \funcfont{FindCordon} in \cref{alg:dp-glws}, 
which attempts to extend the cordon by a batch of $2^{t-1}$ states for increasing $t$ in each substep $t$. 
If the entire batch is ready---i.e., no states in $[\now+1,\now+2^t)$ can be relaxed by each other, and all sentinels are outside the batch---we 
try a larger step and extend the cordon to $\now + 2^{t+1}$. 
During the process, we will maintain $\cordon$ as the leftmost sentinel so far. 
Once we find $\cordon$ is inside the batch, %on the left to the current end of the batch, 
it means that this batch is not fully ready. 
Therefore, the process stops and returns the current value of $\cordon$ to the main algorithm. 
%We first discuss how to avoid checking all $j>\now$. 
%Our 

Using prefix doubling, the parallel algorithm may check more states than the ready ones, %those should be in the frontier,
but the number of ``wasted'' states is at most twice of the 
``useful'' ones which will be finalized in this round. 
%amortized to the number of ``useful'' ones which will be finalized in this round. 
Hence, the total number of processed states is $O(n)$. 

We then discuss the way to avoid checking all states $i>j$ when \stt{$j$} puts sentinels. 
By DM, if $j$ can successfully relax $i$, then $j$ can also successfully relax all states $i..n$.
Therefore, we only need to put a sentinel at the first such state $i$. 
Recall that we maintain all best decision triples in a data structure $B$ in sorted order.
By DM, we can simply binary search ($O(\log n)$ cost) in $B$ to find $s_j$ as the first tentative state that can be updated by \stt{$j$},
and put a sentinel there.
%which takes $O(\log n)$ cost. 

The \funcfont{FindCordon} in \cref{alg:dp-glws} gives the full process as described above.
Each iteration of the while-loop at \cref{line:1d:forloop} is a substep,
which processes a batch of states in $[\now+2^{t-1},\now+2^{t})$ in substep $t$. 
Then for each state $j$ in this batch (in parallel), 
we use $B$ to find the first state that can be updated by $j$ and put a sentinel at this position $s_j$.
Finally, the leftmost sentinel so far forms the cordon. When the cordon is within the current batch, the algorithm returns. 
We also show an illustration of this process in \cref{fig:1d-parallel}. 

\begin{lemma}\label{lemma:findcordon}
  The function \funcfont{FindCordon} has $O(\ffsize \log n)$ work and $O(\log^2 n)$ span, where $\ffsize=\cordon-\now-1$ is the frontier size.
\end{lemma}
\begin{proof}
As discussed above, the prefix doubling scheme 
  may attempt to process up to $\ffsize'$ states, where $\ffsize' \le 2\ffsize$. 
  For each such state, we may binary search in $B$ to find $s_j$ in $O(\log n)$ cost, and check the condition on \cref{line:1d:try-relax} in $O(1)$ cost.
  Therefore, \funcfont{FindCordon} has work $O(\ffsize\log n)$ and span $O(\log^2 n)$.
\end{proof}

\begin{figure}[t]
  \centering
  \includegraphics[width=\columnwidth]{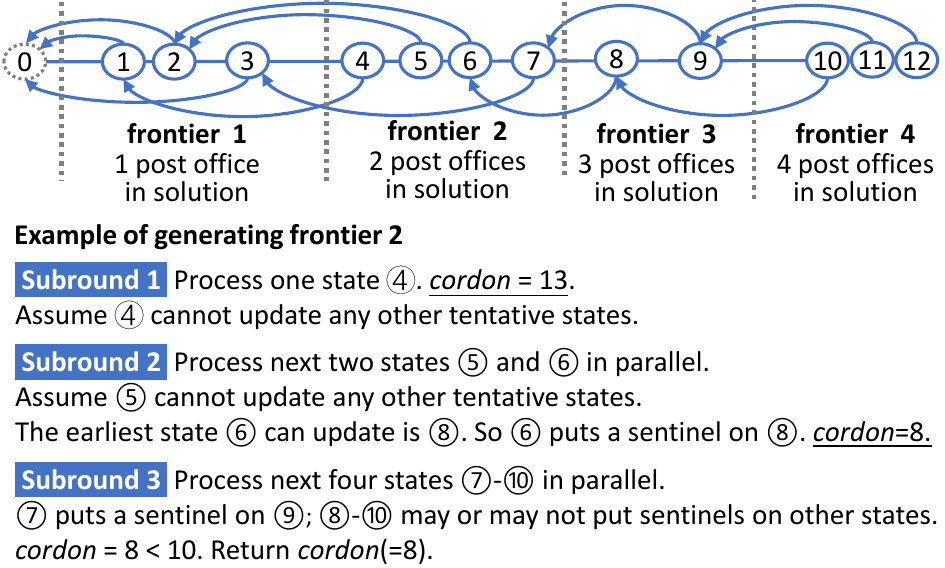}
  %\vspace{.1em}
  \caption{Example of applying the \ouralgo{} to the post office problem with convex cost function. 
  Circles (states) are villages. 
  Arrows are best decisions between states.
  %The red arrows are the best decision found for the entire problem, 
  The final answer is four post offices serving villages 1--3, 4--7, 8--9, 10--12, respectively. 
  The subrounds below illustrate the prefix-doubling scheme in \funcfont{FindCordon}. 
  \label{fig:1d-parallel}}
\end{figure}

\hide{
In the \ouralgo{}, if tentative state $j$ can successfully relax tentative state $i$,
we should put a sentinel on $i$.
Then the cordon will be the minimum position of sentinels.
By decision monotonicity, if $j$ can successfully relax $i$,
then $j$ can also successfully relax any state in $[i..n]$.
So for each $j$, we can binary search the first position $s_j$ that $j$ can relax $s_j$,
and put a sentinel on $s_j$.
As we have stored all best decisions,
we can directly do this binary search on array $B$,
which takes $O(\log n)$ time for each $j$.
}

\hide{
To achieve work efficiency, in each round we need to keep the work of finding the cordon to be proportional to the number of ready states.
So we cannot afford calling the binary search for each $j\in[\now+1,n]$.
$\funcfont{FindCordon}$ shows a dual binary search method.
We keeps doubling the set of ready states
until the minimum sentinel is within the states we have checked,
so all future sentinels will not be smaller than the current one.

Let $L=\nxt-\now-1$,
then $\funcfont{FindCordon}$ has $O(L\log n)$ work and $O(\log^2n)$ span.
}

\subsubsection{Generating New Best-Decision Array}\label{sec:updateB}

\hide{
In the sequential algorithm, after we compute $D[i]$,
we binary search the first index $s_i$ that $i$ can successfully relax $D[s_i..n]$,
and update $\best[s_i..n]$ to $i$.
However, in \cref{alg:dp-new} we have a batch of newly finalized states $D[\now+1..\nxt-1]$,
and their ranges to update can affect each other.
}

The efficiency of the algorithm relies on maintaining an ordered data structure $B$
for all best decision triples. 
We will store $B$ as an array of all such triples in sorted order, such that the binary search 
in \cref{line:1d:try-relax} can be performed efficiently. 
Therefore, after we get the newly finalized states $D[\now+1..\nxt-1]$,
we need to update $B$ accordingly to get the new best decision for all states in $[\cordon, n]$. 

We use a divide-and-conquer approach to do this.
%Function \funcfont{FindIntervals}$(\jl,\jr,\il,\ir)$ uses $D[\jl..\jr]$ to update $\best[\il..\ir]$,
%and find all best decision triples. 
Function \funcfont{FindIntervals}$(\jl,\jr,\il,\ir)$ finds all best decision triples for states in range $[\il,\ir]$, with best decisions
in range $[\jl,\jr]$. 
Note that we only need the best decisions for all states after $\cordon$. 
All these states must have their current best decisions within $[\now+1, \cordon-1]$ 
(if their best decisions are before $\now$, they must have been ready in this round and been included in the frontier). 
Therefore, at the root level, we call \funcfont{FindIntervals}$(\now+1,\cordon-1, \cordon, n)$. 

In \funcfont{FindIntervals}, we first compute $\jm=\best[\im]$ where $\im=(\il+\ir)/2$, i.e., the best decision of the state in the middle.
By (convex) DM, the best decisions of $i\in[\il,\im-1]$ are in $[\jl,\jm]$, 
and the best decisions of $i\in[\im+1,\ir]$ are in $[\jm,\jr]$.
We will deal with the two subproblems in parallel.
To collect all $([l,r],j)$ triples in parallel, we build a tree-based structure bottom-up in the recursion.
%The leaves are the triples obtained by the base cases, and the intermediate nodes connect them into a tree structure. 
%After we find all the intervals $([l,r],j)$,
Finally, we flatten the tree to an array and 
merge the adjacent intervals if they have the same value of $j$. 

\begin{lemma}\label{lemma:updatebest}
  The function \funcfont{UpdateBest} has $O(\ffsize\log n)$ work and $O(\log^2 n)$ span,
  where $\ffsize=\cordon-\now-1$ is the frontier size. 
\end{lemma}
\begin{proof}
  Flattening and removing duplicates can be performed by simple parallel primitives on trees and arrays in $O(\ffsize)$ work and $O(\log n)$ span. 
  Below we will focus on the more complicated \funcfont{FindIntervals} function. 
  The span of \funcfont{FindIntervals} comes from 1) $O(\log n)$ levels recursions and 2) $O(\log n)$ span to check all states in $[\jl,\jr]$ in parallel.
  For the work, each recursive call in \funcfont{FindIntervals} deals with a range of states $[\il,\ir]$ using best decision candidates in range $[\jl,\jr]$.
  The algorithm first finds $\im\in [\il,\ir]$ and its best decision $\jm \in [\jl,\jr]$. 
  This can be done by comparing all possible decisions in $[\jl,\jr]$, which is $O(\jr-\jl)$ work. 
  split the ranges into two subproblems and recurse. Let $N=|\ir-\il+1|$ and $M=|\jr-\jl+1|$ denoting
  the sizes of the two ranges. The work of \funcfont{FindIntervals} indicates the following recurrence:
  $$W(N,M)=W(N/2,M_1)+W(N/2,M_2)+O(M)$$
  where $M_1+M_2=O(M)$. This solves to $O(M\log n)$. On the root level, $M=\cordon-\now-1=\ffsize$. 
  This proves that the work is $O(\ffsize\log n)$ for frontier size $\ffsize$.
\end{proof}

\hide{
Our algorithm tackles this by doing the binary search for all ready states at the same time.
Suppose we want to use $D[j_l..j_r]$ to update $\best[i_l..i_r]$.
We first compute $j_0=\best[i_m]$.
Then by decision monotonicity, the best decisions of $i\in[i_l,i_m-1]$ is in $[j_l,j_0]$,
and the best decisions of $i\in[i_m+1,i_r]$ is in $[j_0,j_r]$.
After we find all the intervals $([l,r],j)$,
we merge the adjancent intervals if they have the same $j$.
}

\hide{
It can be shown that the work of $\funcfont{UpdateBestChoice}$ is proportional to the number of ready states.
The only problem is how to efficiently adding new intervals into $B$ while keeping their order.
As the function $\funcfont{FindBestChoiceIntervals}$ is called recursively in two directions,
we can build a binary tree of intervals based on the recursion structure,
and then flatten the tree into an array.
This method does not bring any race condition.
Another method is to maintain the min/max index for each ready state $j$.
When we add a new interval $([l,r],j)$, we call one $write_{min}$ and one $write_{max}$ for state $j$.
At last we can easily build $B$ based on the range of each $j$.

Let $L=\nxt-\now-1$,
then \cref{alg:update-new} has $O(L\log n)$ work and $O(\log^2n)$ span.
}

\subsection{Parallel Concave GLWS}\label{sec:post-office-concave}

%If $w$ is concave, we can still use the framework of \cref{alg:dp-new} with some modifications. 
%There are two differences with the convex case.
To extend the algorithm to the concave case, we need a few modifications.
In \funcfont{FindCordon}, by the concavity,
if $j$ can update $i$, then $j$ must be able to update $j+1$.
Therefore, in \cref{line:1d:try-relax} in \cref{alg:dp-glws}, we check whether $j$ can update $j+1$. 
If so, we put a \sentinel{} at $j+1$.
The other modifications are in \funcfont{FindIntervals}. 
First, due to concavity,  
when we find $\jm$ as the best decision of $\im$ in \cref{line:1d:jm},
we need to swap the last two parameters in the first and second recursive calls, 
%we need to swap the ranges of decisions in \cref{line:xx} in \cref{alg:dp-new},
%i.e., search the range of states $\im+1$ to $\ir$ in the first recursive call and $\il$ to $\im-1$ in the second. 
i.e., the best decision range for states $\im+1$ to $\ir$ must be before $\jm$, and those in $\il$ to $\im-1$ must be after $\jm$. 

A more involved modification in the concave GLWS is that after we get the array $B$ from \funcfont{FindIntervals},
we have to merge it with the old array $B$ before this round --- \funcfont{FindIntervals} only considers the best decisions among $[\now+1,\cordon-1]$,
but in the concave case, these states may also have better decisions using states before $\now$.
\iffullversion{Suppose we have generated the array $B_\new$ storing $([l,r],j)$ triples,
and we want to merge it with $B_\old$.
Both of $B_\old$ and $B_\new$ contain the best decisions of states $[\cordon..n]$.
The difference is that the $j$s in $B_\old$ are from $[0..\now]$,
while the $j$s in $B_\new$ are from $[\now+1..\cordon-1]$.
By the concave decision monotonicity,
the key is to find a cutting point $p$, where
the best decisions of $[\cordon..p]$ are from $B_\new$,
and the best decisions of $[p+1,n]$ are from $B_\old$.

\begin{algorithm}[ht]
\small
\caption{Find the cutting point $p$\label{alg:mergeB}}
\SetKwProg{myfunc}{Function}{}{}
\SetKwFor{parForEach}{ParallelForEach}{do}{endfor}
\SetKwFor{mystruct}{Struct}{}{}
\SetKwFor{pardo}{In Parallel:}{}{}
\parForEach{$([l_k,r_k],j_k)$ in $B_\new$}{
  $x_k\gets$ search the best decision of $l_k$ from $B_\old$.\\
}
Binary search the last $([l_k,r_k],j_k)$ in $B_\new$ such that
$E[j_k]+w(j_k,l_k)<E[x_k]+w(x_k,l_k)$. \\
Binary search the first $([l_t,r_t],j_t)$ in $B_\old$ such that
$E[j_t]+w(j_t,r_t)<E[j_k]+w(j_k,r_t)$.\\
Binary search the last $p$ in $[l_k,r_t]$ such that
$E[j_k]+w(j_k,p)<E[j_t]+w(j_t,p)$.\\
\Return $p$
\end{algorithm}

Here we show \cref{alg:mergeB} to find $p$ in $O(h\log n)$ work and $O(\log n)$ span,
where $h=|B_\new|$ is the frontier size.
For all $l_k$ in $B_\new$, we pre-process its best decision stored in $B_\old$.
This step requires $O(h\log n)$ work and $O(\log n)$ span.
Then we search in $B_\new$ to find the interval that $p$ locates in.
After this step, there is only one interval $([l_k,r_k],j_k)$ in $B_\new$ is interesting.
Then we can binary search in $B_\old$ to find the exact $p$.
Note that this method can be easily modified to merge $B_\old$ and $B_\new$ even if the cost function is convex.
}\ifconference{For page limitation, we elaborate on this part in full paper~\cite{fullversion}.
With careful design, this part can also be finished in $O(\ffsize\log n)$
work and $O(\log n)$ span, where $\ffsize$ is the frontier size.
}

\hide{
However, this is not a perfect parallelization.
The number of rounds may be larger than the depth of the perfect DAG.
}

\subsection{Theoretical Analysis}

In this section, we show the theoretical analysis for our parallel GLWS algorithm. We first summarize our main results as follows. 

\begin{theorem}\label{thm:1d:convex}
  \hide{
  Given an input sequence of size $n$, the \ouralgo{} for the \emph{convex} GLWS has $O(n\log n)$ work and $O(k\log^2n)$ span,
  where $k$ is the effective depth of the perfect DAG 
  $G^*_{\gammalws}$ of algorithm $\gammalws$ introduced in \cref{sec:post-office:prelim}. 
  It is a \emph{perfect parallelization} of $\gammalws$. 
  }
  \hide{
  Given an input sequence of size $n$, and the sequential GLWS algorithm 
  $\gammalws$ introduced in \cref{sec:post-office:prelim},   
  the \ouralgo{} for the \emph{convex} GLWS has $O(n\log n)$ work and $O(k\log^2n)$ span,
  where $k$ is the effective depth of the perfect DAG 
  $G^*_{\gammalws}$. 
  It is a \emph{perfect parallelization} of $\gammalws$. 
  }
  Given an input sequence of size $n$, and the sequential GLWS algorithm 
  $\gammalws$ introduced in \cref{sec:post-office:prelim},   
  let $k=\ed(G^*_{\gammalws})$ be the effective depth of the 
  $\gammalws$-perfect DAG. 
  Then the \ouralgo{} for the \emph{convex} GLWS has $O(n\log n)$ work and $O(k\log^2n)$ span. 
  It is a \emph{perfect parallelization} of $\gammalws$.  
\end{theorem}

More intuitively, $k$ in \cref{thm:1d:convex} is also the number of best decisions to make in the final solution: 
for instance, for the post-office problem, it is the number of post offices in the optimal solution. 

\begin{theorem}\label{thm:1d:concave}
\hide{
  The \ouralgo{} for the \emph{concave} GLWS has $O(n\log n)$ work and $O(k\log^2n)$ span,
  where $n$ is the input size and $k$ is the effective depth of the optimized DAG $G_{\gammalws}$
  algorithm introduced in \cref{sec:post-office:prelim}. 
  It is an \emph{optimal parallelization} of $\gammalws$.   }
  \hide{
  Given an input sequence of size $n$, and the sequential GLWS algorithm 
  $\gammalws$ introduced in \cref{sec:post-office:prelim},
  the \ouralgo{} for the \emph{concave} GLWS has $O(n\log n)$ work and $O(k\log^2n)$ span,
  where $k$ is the effective depth of the optimized DAG $G_{\gammalws}$
  algorithm introduced in \cref{sec:post-office:prelim}. 
  It is an \emph{optimal parallelization} of $\gammalws$.    }
  
  Given an input sequence of size $n$, and the sequential GLWS algorithm 
  $\gammalws$ introduced in \cref{sec:post-office:prelim},   
  let $k=\ed(G_{\gammalws})$ be the effective depth of 
  %the \emph{optimal DAG}  $G_{\gammalws}$. 
  the $\gammalws$-optimized DAG. 
  Then the \ouralgo{} for the \emph{concave} GLWS has $O(n\log n)$ work and $O(k\log^2n)$ span. 
  It is an \emph{optimal parallelization} of $\gammalws$.  
\end{theorem}

We first prove that both algorithms are nearly work-efficient and have $O(n\log n)$ work.

\begin{lemma}
  The \ouralgo{} for GLWS has $O(n\log n)$ work for both convex and concave case.
\end{lemma}
\begin{proof}
  Combining \cref{lemma:findcordon,lemma:updatebest} (and the discussion in \cref{sec:post-office-concave}), the work for each round is $O(\ffsize\log n)$, 
  where $\ffsize=\cordon-\now-1$ is the frontier size. 
  Since the frontier sizes $\ffsize$ across all rounds add up to $n$, 
  the entire algorithm has $O(n\log n)$ work. 
  %\cref{line:1d:compute} computes all states in the frontier by searching their best decisions in $B$, which also costs $O(\ffsize\log n)$ work. 
  %Finally, the cost of \funcfont{UpdateBest}, as shown in \cref{lemma:1d:updateBcost} and \cref{sec:1d:concave}, also have $O(\ffsize\log n)$ work.
\end{proof}

We then show that the number of rounds in both convex and concave cases is the effective depth of $G_{\gammalws}$.
Recall that the DAG $G_{\gammalws}$ includes normal edges between all states $j$ and $i<j$,
and effective edges between a state $j$ and its best decision. 
The effective depth $\ed(G_{\gammalws})$ is the largest number of effective edges in any path. 

\begin{lemma}\label{lemma:1d:effectivedepth}
  The \ouralgo{} for GLWS finishes in $k$ rounds, where $k$ is $\ed(G_{\gammalws})$. 
\end{lemma}
\begin{proof}
  Define the effective depth $\ed(s)$ of a state $s$
  as the largest number of effective edges of a path ending at $s$. 
  We will inductively prove that a state $s$ is in the frontier of round $r$ iff. $s$ has effective depth $r$.
  The base case (boundary cases) holds trivially.

  Assume the conclusion is true for $r-1$. We first prove the ``if'' direction, i.e., if a state $s$ has effective depth $r$,
  it must be in the frontier of round $r$. This is equivalent to show that there is no sentinel on all states from $\now$ to $s$.
  Assume to the contrary that there is a state $y\in (\now,s]$ with a sentinel, which is put by state $x\in (\now, y]$.
  This means that $x$ is a better decision for $y$ than all states before $\now$, 
  indicating that $y$'s best decision $y^*\ge\now$.
  %Therefore, the best decision of $y$, denoted as $y^*$, must have $y^*\ge \now$.
  Based on the induction hypothesis, the effective depth of $y^*$ must be larger than $r-1$. 
  Therefore, $\ed(y)= \ed(y^*)+1 >r-1+1=r$, which means that $\ed(y)$ is at least $r+1$.
  Based on the recurrence, there is a normal edge from $y$ to $s$, so $\ed(s)\ge r+1$, leading to a contradiction.
  
  We then prove the ``only if'' condition, i.e., if a state $s$ is in the frontier of round $r$, it must have effective depth $r$. 
  The induction hypothesis suggests that all states with effective depth smaller than $r$ have been finalized in previous rounds, so we only
  need to show that $\ed(s)$ cannot be larger than $r$. 
  Assume to the contrary that $\ed(s) \ge r+1$. 
  Let the path to $s$ with effective depth $\ed(s)$ be $x_1, x_2, \dots, s$. 
  Since the total number of effective edges on this path is at least $r+1$, there must exist
  an effective edge $\edge{x_i}{x_{i+1}}$ on the path such that $\ed(x_i)=r$ and $\ed(x_{i+1})=r+1$.
  However, based on the induction hypothesis, $x_i$'s best decision must have been finalized. 
  During \cref{line:1d:computeD}, $x$ must get its true DP value, and will find itself able to update $x_{i+1}$.
  Therefore, there will be a sentinel on $x_{i+1}\le s$, and $s$ cannot be identified in the frontier of round $r$.
\end{proof}

We will then show that the number of rounds of the convex case is also the effective depth of the 
\emp{$\gammalws$-perfect DAG} $G^*_{\gammalws}$. 
This is stronger than the $\gammalws$-optimized DAG as shown above. 
Recall that the perfect DAG $G^*_{\gammalws}$ contains all best decision edges in $G_{\gammalws}$.

\begin{lemma}
  \label{lemma:1d:perfectdepth}
  %The perfect depth of $\gammalws$ is the same as the effective depth for convex case. 
  The \ouralgo{} for convex GLWS runs in $k^*$ rounds, where $k^*$ is $\ed(G^*_{\gammalws})$. 
\end{lemma}
\begin{proof}
  Define the perfect depth $\ds(s)$ of a state $s$
  as the largest number of effective edges of any path ending at $s$ in $G^*_{\gammalws}$.   
  Similarly, we will show by induction that in round $r$, all states with perfect depth $r$ will be processed. The base case holds trivially. 
  Assume the conclusion holds for $r-1$. 
  In round $r$, we will show that a state $s$ with perfect depth $r$ must be put in the frontier. %because no sentinels will be put between $\now$ and $s$.
  Let $s^*$ be the best decision of $s$, then $\ds(s^*)=r-1$ and therefore $s^*<\now$. According to DM, any state $x$ between $\now$ and $s$ must have its best decisions
  $x^*\le s^*<\now$, indicating that $\ds(x^*)\le r-1$. Therefore, $x$ must find its true best decision in $B$, and cannot be updated by any other tentative states in \cref{line:1d:try-relax}. This means that there will be no sentinel between $\now$ and $s$, so $s$ must be identified ready in round $r$. 
  Therefore, a state with perfect depth $r$ must be finalized in round $r$, leading to the stated theorem. 
\end{proof}

%Note that the span of each round is $O(\log^2 n)$, caused by the \funcfont{UpdateBest} function (see \cref{thm:updatebest}).
Combining \cref{lemma:findcordon,lemma:updatebest,lemma:1d:effectivedepth,lemma:1d:perfectdepth}
%\xiangyun{why this is Thm. instead of Lem.?} 
proves the span bounds in \cref{thm:1d:convex,thm:1d:concave}.

\hide{
\myparagraph{Discussions with totally monotonicity.}
A more general condition to enable DM is \emph{convex/concave totally monotonicity}~\cite{aggarwal1987geometric}.
Our algorithm works as long as the transition function is convex/concave totally monotone. 
For page limit, we formally explain this in \iffullversion{\cref{app:dm}}\ifconference{full version of this paper \cite{fullversion}}.
}

\hide{
\begin{theorem}
  The \ouralgo{} for the convex GLWS has $O(n\log n)$ work and $O(k\log^2n)$ span,
  where $n$ is the problem size and $k$ is the longest dependency path. It is a perfect parallelization of the $\Gamma_{1D}$
  algorithm introduced in \cref{sec:1d:prelim} for the convex case. 
\end{theorem}
}

\section{Other Parallel DP Algorithms}\label{sec:other}

We now show that our algorithmic framework can be used to parallelize a wide variety of classic sequential DP algorithms.
\iffullversion{}
\ifconference{Due to the space limit, we will focus on two of them and leave others in \iffullversion{the appendices}\ifconference{the full version of this paper \cite{fullversion}}, with a short summary in \cref{sec:other-other}.

}
In particular, for the optimal alphabetic tree (OAT) problem (\cref{sec:oat}), we partially answered a long-standing open problem by \citet{larmore1993parallel} for reasonable input instances (for instance, positive integer weights in range $n^{\polylog(n)}$).
For the GAP problem (\cref{sec:gap}), we showed the first nearly work-efficient algorithm with non-trivial parallelism.
More interestingly, this algorithm combines all techniques in the algorithms for convex GLWS and sparse LCS. 
%and it might be the most algorithmically interesting one in this paper.
%We also show tree-based DP structure with DM in \cref{sec:other-other}, which requires a list of algorithmic components such as the heavy-light decomposition, range trees, etc.

\subsection{Parallel Optimal Alphabetic Trees (OAT)}\label{sec:oat}

The optimal alphabetic tree (OAT) problem is a classic problem
and has been widely studied both sequentially \cite{karpinski1997correctness,itai1976optimal,van1976construction,nagaraj1997optimal,hu1971optimal,garsia1977new,davis1998hu,larmore1998optimal}
and in parallel
\cite{larmore1996parallel,rytter1988efficient,larmore1993parallel,larmore1995constructing}.
Given a sequence of non-negative weights $a_{1..n}$, the OAT is a binary search tree with $n$ leaves and has the minimum cost,
where the cost of a tree $T$ is defined as:
\begin{equation}
  cost(T)=\sum_{i=1}^na_id_i
\end{equation}
Here $d_i$ is the depth of the $i$-th leaf (from the left) of $T$ (the root has depth 0).
One can view the weight $a_i$ as the frequency of accessing leaf $i$, and the depth of a leaf is the cost of accessing it.
Then the cost of $T$ is the total expected cost of accessing all leaves in $T$. 
%The OAT problem is a special case of the optimal binary search tree (OBST) problem where all the weights are in the leaves,
%and it is closely related to the Huffman coding~\cite{huffman1952method}.
The OAT problem is closely related to other important problems such as the optimal binary search tree (OBST) \cite{knuth1971optimum}
and Huffman tree~\cite{huffman1952method}. 
%For example, a Huffman tree is also a search tree with minimum cost for an input weight sequence, 
%but allows for shuffle all elements $a_i$. 

Sequentially, \citet{hu1971optimal} showed an OAT algorithm with $O(n\log n)$ work. 
Later, \citet{garsia1977new} simplified this algorithm. 
In parallel, \citet{larmore1993parallel} showed an algorithm based on Garsia-Wachs.
We will apply our techniques to this algorithm to improve the span bounds. 
Due to page limit, we provide the details of ~\cite{larmore1993parallel} in \iffullversion{\cref{sec:app-oat}}\ifconference{the full version paper \cite{fullversion}}, and review the high-level idea here.
The algorithm computes an $l$-tree~\cite{garsia1977new},
which has the same depth with and will be finally converted to the OAT in $O(n)$ work and polylogarithmic span. 
The key insight of ~\cite{larmore1993parallel} is to start with a sequence of $n$ leaf nodes with the input weights, 
and find several disjoint intervals in the sequence to process in parallel. 
This partition is done by various operations on the Cartesian tree of the input sequence, which requires $O(n\log n)$ work and $O(\log^2 n)$ span.
Larmore et al.\ showed that processing each interval can be reduced to a convex LWS. 
The solution of the LWS will connect items in this interval into a forest, 
%which outputs a forest for all elements in this interval. 
%and the output is a forest based on the best decisions of the LWS problem. 
%Each forest will be converted into an \emph{$l$-tree}, which corresponds to (and will finally be converted to) a subtree in OAT. 
which becomes a subgraph in the final $l$-tree. 
%Generating an $l$-tree for an interval with size $m$ takes $O(m\log m)$ work and $O(\log^2 m)$ span.
%For an interval of size $m$, one can postprocess the forest in $O(m\log m)$ work and $O(\log^2 m)$ span
%such that it can be finally converted to the components in OAT. 
%we can convert the forest to a subgraph in the final OAT. 
%Each tree in the forest corresponds to a subgraph of the final OAT. 
%Using an algorithm with $O(m\log m)$ work and $O(\log^2 m)$ span, the forest can be converted to a subtree in L-tree, 
%where $m$ is the length of this interval.
%Finally, we insert the total weight of each tree back to the original sequence, this new 
Finally, for each tree in the forest, we insert its root back to the sequence and repeat the process. 
This reinsertion step takes $O(n\log n)$ work and $O(\log n)$ span by basic parallel primitives such as sorting and range-minimum queries. 
%We repeat this process on the sequence with newly-inserted tree nodes, which further connects the forests to the final $l$-tree. %until there is only one tree left. 
The further rounds will connect the forest to the final $l$-tree. 
%Throughout the algorithm, the 
%\yihan{I'm still unsure about why the best decision tree has the same depth as a component in the $l$-tree. The best decision tree seems to be a multi-way tree but the
%final OAT tree should be a binary tree. }
Larmore et al.\ also showed that the number of such intervals shrinks to half in each iteration, so the algorithm will finish in $O(\log n)$ rounds. 
Here all other steps in addition to solving convex LWS take $O(n\log^2 n)$ work and $O(\log^3 n)$ span. 
The remaining cost of the algorithm is to solve convex LWS in each round, multiplied by the number of rounds, which is $O(\log n)$.

Larmore et al.\ originally used the parallel convex LWS algorithm from~\cite{apostolico1990efficient,aloknotes}, 
which has $O(m^2\log m)$ work and $O(\log^2m)$ span when taking an input interval of length $m$. 
%Since the total length for all intervals in one recursive iteration is at most $n$, in the worst case when one interval has size $O(n)$, the bounds are dominated by this single subproblem.
%Hence, the entire algorithm requires $O(n^2\log^2n)$ work and $O(\log^3n)$ span.
Later, \citet{larmore1995constructing} improved the parallel convex LWS algorithm to $O(m^{1.5}\log m)$ work and $O(\sqrt{m}\log m)$ span,
yielding $O(n^{1.5}\log^2n)$ work and $O(\sqrt{n}\log^2n)$ span for the OAT algorithm---the work overhead is still polynomial.
%However, the work overhead is still $\tilde{O}(\sqrt{n})$, making this algorithm mostly of theoretical interest.
%However, the work overhead is still $\tilde{O}(\sqrt{n})$. 
%Despite a polylogarithmic span bound, 
\citet{larmore1993parallel} left the open problem on whether there exists an OAT algorithm with $\tilde{O}(n)$ work and $\polylog(n)$ span, 
which remains unsolved for three decades.

Note that the convex LWS problem is a special case of the convex GLWS problem discussed in \cref{sec:post-office:prelim} with $E[i]=D[i]$.
Hence, \cref{alg:dp-glws} directly gives $O(m\log m)$ work and $O(k\log^2m)$ span for convex LWS problem, and here $k$ is the longest dependency path of best decisions.
%We note that since each forest corresponds to a subcomponent in the final OAT tree,
%One useful conclusion in Larmore's algorithm is that, 
%in each LWS subroutine, 
%if the best decision for element $x$ is $y$, then the $x$ must be one level deeper than $y$ in the final OAT. 
In Larmore's algorithm, the forest for each interval is constructed iteratively by the DP algorithm on LWS: 
if iteration $i$ finds the best decision at iteration $j<i$, 
then iteration $i$ creates one more level on top of the forest at iteration $j$.
This means that the $k$ is equivalent to the depth of the forest, which is upper bounded by the final OAT height $h$. 
We present more details in \iffullversion{\cref{sec:app-oat}}\ifconference{\cite{fullversion}}.
%Note that here in this approach, the output size $k$ corresponds to the number of levels of the generated OAT.
%Hence, $k$ is upper bounded by the total height $h$ of the output OAT, and we have the following theorem.
Hence, we can parameterize our final bounds using $h$ as:

\begin{theorem}\label{theorem:oat}
  The optimal alphabetic tree (OAT) can be constructed in $O(n\log^2n)$ work and $O(h\log^3n)$ span,
  where $n$ is the size of input weight sequence and $h$ is the height of the OAT.
\end{theorem}

This algorithm is nearly work-efficient with span parameterized on $h$.
One useful observation is that the OAT height $h$ is polylogarithmic with real-world input instance with positive integer weights and fixed word length.
%(this also hold for similar tree structures such as Huffman trees)
More formally, we can show that:

\hide{
\begin{lemma}\label{lem:oat-height-lemma}
  In an OAT, the subtree weight grows by at least twice for every three levels.
\end{lemma}
}

\begin{lemma}\label{lem:oat-height-lemma}
  If all input weights are positive integers in word size $W$, 
  the OAT height is $O(\log W)$. 
\end{lemma}

The proof is not complicated and given in \iffullversion{\cref{sec:app-oat}}\ifconference{the full paper~\cite{fullversion}}.
With this lemma, we can state the following corollary.

\begin{corollary}\label{cor:oat}
  If the input key weights are positive integers with word size $W=n^{\polylog n}$,
  the OAT can be constructed in $O(n\log^2n)$ work and $\polylog(n)$ span,
  where $n$ is the input size.
\end{corollary}

%This can be proved by observing that the total weight is no more than $Wn$, 
%we have the tree height $h=O(\log Wn)=\polylog(n)$.
The bounds also hold for real number weights if the ratio between the largest and smallest weight is $n^{\polylog n}$.
We note that in realistic models we usually assume word-size $W=n^{O(1)}$, %which is the case for real computers.
%In these realistic cases, we affirmatively answers the open problem in \cite{larmore1993parallel}.
in which case \cref{cor:oat} affirmatively answers the open problem in \cite{larmore1993parallel}.

\subsection{The GAP Edit Distance Problem}\label{sec:gap}

%The GAP Edit Distance (GAP) problem~\cite{}
%is a combination of LCS and the GLWS problem.
%is a 2D version of the GLWS problem.
%Given two strings $A[1..m]$ and $B[1..n]$,
%a sequence of deletes/inserts correspond to a gap in $A$ or $B$, respectively.
The GAP problem is a variant of the famous edit distance problem.
The GAP problem aligns two input strings with sizes $n$ and $m\le n$, and allows editing a substring with certain cost function (formally defined below).
This problem has been widely studied both sequentially~\cite{galil1989speeding,eppstein1988speeding,chowdhury2006cache}
and in parallel~\cite{BG2020,chowdhury2010cache,galil1994parallel,tang2017,chowdhury2016autogen,itzhaky2016deriving,tithi2015high}.
% \yan{(add [34,36,31,60,77,74] from my paper, check if they are sequential or parallel).}
As noted by \citet{eppstein1988speeding}, most real-world cost functions are either convex or concave, yielding $\tilde{O}(nm)$ work for the GAP problem sequentially.
Unfortunately, to the best of our knowledge, these existing parallel algorithms for the GAP problem need $\Omega(n^2m)$ work,
and the polynomial overhead makes them less practical.

More specifically, GAP takes two strings $A[1..n]$ and $B[1..m]$, and computes
the minimum cost to align $A$ and $B$ using the following operations:
1) deleting $A[l+1..r]$ with cost $w_1(l,r)$,
and 2) deleting $B[l+1..r]$ with cost $w_2(l,r)$.
Here we consider the following recurrence, which is usually referred to as the GAP recurrence:
\begin{align*}
P[i,j]&=\min_{0\le i'<i}D[i',j]+w_1(i',i) \\
Q[i,j]&=\min_{0\le j'<j}D[i,j']+w_2(j',j) \\
D[i,j]&=\min\{P[i,j],Q[i,j],~D[i-1,j-1]~|~A[i]=B[j]\}.
\end{align*}
Here $P[i,j]$ and $Q[i,j]$ indicate the edits on the two strings.
Directly computing the recurrence uses $O(n^2m)$ work.
Since most real-world cost functions in machine learning, NLP, and bioinfomatics~\cite{eppstein1988speeding} are either convex or concave, 
sequentially each row in $P$ or column in $Q$ is a convex or concave GLWS and can be computed in $\tilde{O}(n)$ or $\tilde{O}(m)$ work.
Hence, computing the entire $P$ and $Q$ takes $\tilde{O}(nm)$ work, leading to the same cost for computing $D$ and the entire problem. 
We denote this standard sequential algorithm as $\gammagap$. 

\begin{figure}[t]
  \centering
  \includegraphics[width=0.75\columnwidth]{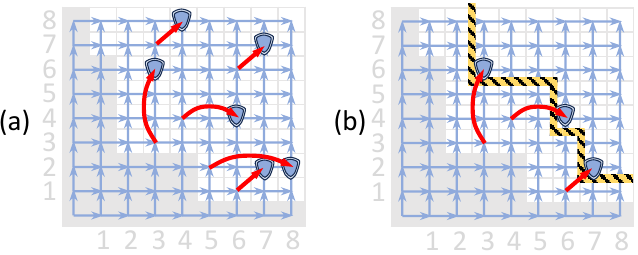}
  \caption{Example of a cordon in the GAP problem.
  \label{fig:2d-gap}}

\end{figure}

Parallelizing this approach is extremely challenging even with the parallel convex/concave GLWS in \cref{sec:post-office} as a subroutine, and we are unaware of any existing work on this.
The challenge here is that the rows in $P$ interact with the columns in $Q$. %are interactive.
For instance, computing a row in $P$ requires one element from each column in $Q$, but computing those elements again requires previous rows in $P$.

Our key insight to parallelize this algorithm is to use the \ouralgo{} to efficiently mark the ready region to be computed in each round.
Note that as a generalization of the classic edit distance/LCS, the GAP recurrence is similar to Recurrence~\ref{eqn:lcs}, but with ``jumps'' in computing $P$ and $Q$.
An illustration is given in \cref{fig:2d-gap}.
In addition to the diagonal edges as in LCS (see \cref{fig:lis}), for rows and columns, 
there also exist effective (red) edges (see \cref{fig:2d-gap}(a)). 
Here for simplicity we only draw a subset of these edges, 
and every state $D[i,j]$ always have one vertical effective edge (to compute $Q[i,j]$), one horizontal effective edge (to compute $P[i,j]$), and may have a diagonal edge if $A[i]=B[j]$.
All these edges imply the sentinels, which form the cordon and imply
the regions for ready states, as shown in \cref{fig:2d-gap}(b). 
The cordon is still a staircase as in LCS. 

However, finding the cordon in GAP is sophisticated.
We cannot directly use a tournament tree as in LIS, since the vertical and horizontal edges are computed on-the-fly and not known ahead of time.
Meanwhile, in a 2D table where the cordon is a staircase, we cannot simply use prefix-doubling as in GLWS in \cref{sec:post-office}.
We propose a unique solution here to use prefix-doubling on a 2D table and computes the staircase cordon efficiently.
This approach will consider each row separately, but for all rows, we run prefix doubling synchronously and try to see if the next ranges are available.
First, we put a \sentinel{} on state $(x,y)$ with a diagonal edge if $(x-1,y-1)$ is not finalized. 
We will maintain the best-decision structure for each row and column, in the same way as the GLWS algorithm. 
For this region to be checked, we will use the same approach as in \cref{alg:dp-glws} to compute $P$ and $Q$, take the minimum as $D$, and use $D$ to check their readiness.
If a state $(x,y)$ obtains the best decision from another tentative state, we will put a \sentinel{} on $(x,y)$, 
which will block the other states $(x',y')$ with $x'\ge x$ and $y'\ge y$.
% An interesting observation here is that since the cordons are staircases, when all rows running prefix doubling, the settled region and checking region must be connected, guaranteeing the feasibility for the checking on rows and columns.
The work to put the \sentinel{}s is proportional to the number of states we checked in the prefix-doubling, and the span is polylogarithmic.

Finally we discuss how to handle the \sentinel{s} placed as above.
We store all \sentinel{s} based on the row index on increasing order.
After this, applying a prefix-min on these \sentinel{s} gives part of the cordon (if they exist), and we will merge it with the previous cordon.
Then, for all tentative states, we check whether they are on the correct side, and invalidate those across the cordon. 
Since we are using prefix doubling, the wasted work for the invalid states can be amortized.
In the next prefix doubling step, we will also use the cordon to limit the search region.
Once all states within the cordon are checked for readiness, we can move to the next round.
Due to prefix doubling, we only need $O(\log n)$ steps in each round.

\begin{theorem}\label{thm:gap}
  The \ouralgo{} for the GAP problem has $O(mn\log n)$ work and $O(k\log^2n)$ span,
  where $n$ and $m\le n$ are the input size and $k$ is the effective depth of the $\gammagap$-optimized DAG
  for the sequential algorithm $\gammagap$ introduced in \cref{sec:gap}.
\end{theorem}

Recall that the sequential GAP algorithm $\gammagap$ gets the DP value
for each state $s=(i,j)$ by solving the GLWS problems in row $i$
and column $j$, respectively, and the diagonal edge $\edge{(i-1,j-1)}{(i,j)}$ if applicable.
Therefore, the optimal DAG $G_{\gammagap}$ contains three types of edges 
\begin{itemize}
  \item $\edge{(i,j)}{(i',j)}$ for all $i'>i$,
  \item $\edge{(i,j)}{(i,j')}$ for all $j'>j$, and
  \item $\edge{(i-1,j-1)}{(i,j)}$ if $A[i]=B[j]$. 
\end{itemize}
Among them, the effective edges include:
\begin{itemize}
  \item $\edge{(i,j)}{(i',j)}$ where $i$ is the best decision for $i'$ in the GLWS problem on row $j$,
  \item $\edge{(i,j)}{(i,j')}$ where $j$ is the best decision for $j'$ in the GLWS problem on column $i$, and
  \item $\edge{(i-1,j-1)}{(i,j)}$ if $A[i]=B[j]$. 
\end{itemize} 

WLOG we assume $m\le n$ in this section. We first prove the span bound. We will show that the \ouralgo{} finishes in $k$ rounds, where $k$ is the effective depth of $G_{\gammagap}$. 
\begin{lemma}\label{lemma:gap-round}
  Given two sequences of sizes $n$ and $m\le n$,
  the \ouralgo{} on GAP edit distance finishes in $k=\ed(G_{\gammagap})$ rounds. 
\end{lemma}
\begin{proof}
  The proof is similar to \cref{lemma:1d:effectivedepth}. We also define the effective depth of a state $s$ as $\ed{s}$.
  We will show by induction that $s$ is processed in round $r$ iff. $\ed{s}=r$. The base case holds trivially. 
  
  Assume the conclusion holds for all rounds up to $r-1$. We will show it is also true for round $r$. 
  We first prove the ``if'' direction, i.e., if a state $s=(i,j)$ ($i$-th row, $j$-th column) has effective depth $r$,
  it must be in the frontier of round $r$. 
  This is equivalent to show that there is no sentinel that blocks $s$.
  For simple description, for two states $s=(i,j)$ and $s'=(i',j')$,
  we say $s\prec s'$ if $i\le i'$ and $j\le j'$. 
  Clearly, if a state $s\prec s'$, a sentinel on $s$ will block $s'$. 
  Assume to the contrary that there is a state $y\prec s$ with a sentinel, which is put by another tentative state $x\prec y$.
  This means that the tentative state $x$ is a better decision than all finalized states, indicating that the best decision of $y$, denoted as $y^*$, must 
  also be tentative.
  %Therefore, the best decision of $y$, denoted as $y^*$, must have $y^*\ge \now$.
  Based on the induction hypothesis, the effective depth of $y^*$ must be larger than $r-1$. 
  Therefore, $\ed(y)\ge \ed(y^*)+1 >r-1+1=r$, which means that $\ed(y)$ is at least $r+1$.
  %Based on the recurrence, there is a normal edge from $y$ to $s$, so $\ed(s)\ge r+1$, leading to a contradiction.
  Let $y=(i',j')$. $y$ and $s$ can be connected by either one normal edge (when they are in the same row or column) 
  or two normal edges ($\edge{(i',j')}{(i,j')}$ and $\edge{(i,j')}{(i,j)}$). 
  This means that the effective depth of $s$ is at least the same as $y$, which is $r+1$. 
  This leads to a contradiction. 
  
  We then prove the ``only if'' condition, i.e., if a state $s$ is in the frontier of round $r$, it must have effective depth $r$. 
  The induction hypothesis suggests that all states with effective depth smaller than $r$ have been finalized in previous rounds, so we only
  need to show that $\ed(s)$ cannot be larger than $r$. 
  Assume to the contrary that $\ed(s) \ge r+1$. 
  Let the path to $s$ with effective depth $\ed(s)$ be $x_1, x_2, \dots, s$. 
  Since the total number of effective edges on this path is at least $r+1$, there must exist
  an effective edge $\edge{x_i}{x_{i+1}}$ on the path such that $\ed(x_i)=r$ and $\ed(x_{i+1})=r+1$.
  However, based on induction hypothesis, $x_i$'s best decision must have been finalized. 
  In round $r$, $x$ must get its true DP value, and will find itself able to update $x_{i+1}$.
  Therefore, there will be a sentinel on $x_{i+1}\prec s$, and $s$ cannot be identified in the frontier of round $r$.
\end{proof}

Combining \cref{lemma:findcordon,lemma:updatebest}, the span in each round is $O(\log^2n)$.
This proves the span bound in \cref{thm:gap}. 

We then prove the work bound in \cref{thm:gap}. 
\begin{lemma}
  Given two sequences of sizes $n$ and $m\le n$,
  the \ouralgo{} on GAP edit distance has work $O(mn\log n)$. 
\end{lemma}
\begin{proof}
  As $G_{\gammagap}$ is a grid graph, its depth is no more than $m+n$.
  By \cref{lemma:gap-round} the algorithm will finish in $k=O(n)$ rounds.
  In each round, we do prefix-doubling across all $m$ rows and try to push the frontier on each row.
  In each prefix-doubling step we do a prefix-min that costs $O(m)$ work,
  so the cost of prefix-doubling is $O(m\log n)$ in each round, and $O(mn\log n)$ in total.
  Suppose $h$ is the frontier size in one round.
  Due to prefix-doubling, the number of tentative states we visited is at most $2h$.
  Combining \cref{lemma:findcordon,lemma:updatebest},
  in each row/column we can achieve work proportional to the number of tentative states.
  Thus the cost to put sentinels and maintain the best decision arrays is also $O(mn\log n)$.
\end{proof}

\ifconference{
\subsection{Other Algorithms}\label{sec:other-other}

%In addition, we also discuss parallelizing the classic sequential DP algorithms for $k$-LWS and optimal binary-search tree (OBST).
%In these cases, \ouralgo{} will give the trivial parallel solutions.

We further apply the \ouralgo{} to several other problems in \iffullversion{\cref{sec:app-tree,sec:k-lws,sec:obst}}\ifconference{the full version of this paper \cite{fullversion}}, including 
a GLWS problem on a tree structure, 
$k$-GLWS, and Optimal Binary Search Tree (OBST). 
For both $k$-GLWS and OBST, our algorithm leads to a correct algorithm with trivial parallelism ($\tilde{O}(k)$ and $\tilde{O}(n)$ span respectively). 
However, we show that they are still optimal parallelization of the sequential algorithms---this indicates that to further improve the parallelism, 
one should probably find or redesign another sequential algorithm to start with.
%However, we will show that they are still optimal parallelizations to the corresponding sequential algorithm. 
%This indicates that achieving non-trivial span may need new insights to redesign the dependencies. 
}

\iffullversion{
\subsection{General LWS on Trees}\label{sec:app-tree}

The idea of decision monotonicity (DM) can be applied to various structures more than just 1D cases discussed in \cref{sec:post-office}.
The efficient parallelism on 2D grid structure is introduced in \cref{sec:gap}, and we now show the techniques to enable high parallelism on the tree structure.
Here we refer to this problem as \treelws{}.

%In the writing of this paper we discovered a new but interesting problem and we call it GLWS on trees.
Let $T$ be a tree with $n+1$ nodes, and node $0$ is the root.
We use $p(v)$ to denote the parent of node $v$ and $d_v$ be the distance from node $0$ to node $v$.
\treelws{} takes the input tree $T$, a cost function $w$, and the boundary $D[0]$, and computes:
\begin{equation}\label{eqn:tree}
D[v]=\min\{E[u]+w(d_u,d_v)\}
\end{equation}
where $u$ is any ancestor of $v$, and $E[u]=f(D[u],u)$ that can be computed in constant time from $D[u]$ and $u$.
The cost function $w$ is decided by the depths of $u$ and $v$.
Note that here sibling nodes $v_1$ and $v_2$ will have the same DP value,
but $E[v_1]$ and $E[v_2]$ can be different given that $v_1$ and $v_2$ are also part of the parameter in computing the function $f$.
%, thus the DP values in the two subtrees will diverge.
In this section we assume $w$ is convex, but our algorithm can adapt to the concave case with some modifications.

\subsubsection{Building Blocks}\label{sec:tree-bb}
We will first overview some basic building blocks, which are crucial subroutines used in our algorithm.
%For better understanding, we first introduce some building blocks that is used to solve this problem.

\myparagraph{Persistent Data Structures.}
A persistent data structure~\cite{persistence} keeps history versions when being modified.
We can achieve persistence for binary search trees (BSTs) efficiently by path-copying \cite{blelloch2016just,sun2018pam,blelloch2022joinable}, where only the affected path related to the update is copied.
Hence, the BST operations can achieve persistence with the same asymptotical work and span bounds as the mutable counterpart.

\myparagraph{Heavy-Light Decomposition (HLD).}
HLD~\cite{sleator1983data} is a technique to decompose a rooted tree into a set of disjoint chains. %, and any tree path can be decomposed into $O(\log n)$ of these chains.
In HLD, each non-leaf node selects one \defn{heavy edge},
the edge to the child that has the largest number of nodes in its subtree.
Any non-heavy edge is a \defn{light edge}.
If we drop all light edges, the tree is decomposed into a set of top-down chains with heavy edges.
As such, HLD guarantees that the path from the root to any node $v$ contains $O(\log n)$ distinct chains plus $O(\log n)$ light edges.
If we use BSTs to maintain each heavy chain in HLD, we can answer path queries (e.g., query the minimum weighted node on a tree path) in $O(\log^2n)$ work.

\myparagraph{Range Report Based on Tree Depth.}
We now discuss a data structure that efficiently reports the set of nodes in a subtree of $T$ where the depths of the nodes are in a given range $l$ to $r$.
First we build the Euler-tour (ET) sequence of $T$, so any subtree of $T$ will be a consecutive subsequence in the ET.
We can map all nodes to a 2D plane each with coordinates $(f_v,d_v)$, where $f_v$ is the first index of $v$ in the ET, and the $d_v$ is the tree depth of $v$.
Now the original query is a 2D range report on this 2D plan.
A range tree~\cite{sun2019parallel} can be built in $O(n\log n)$ work and $O(\log^2 n)$ span, and answer this query in $O(m+\log^2n)$ work and $O(\log^2n)$ span where $m$ is the output size.

\subsubsection{Our Main Algorithm}

Here if we consider any tree path, Recurrence~(\ref{eqn:tree}) is exactly the same as for the 1D case in \cref{sec:post-office}.
Hence, we can use a similar approach as in \cref{sec:post-office} by maintaining the best-decision array of $([l,r],j)$ triples,
meaning that for elements $v$ with depth $l\le d_v\le r$, $v$'s best decision is $j$.
The challenge here is the branching nature of a tree---we need to handle path divergences at nodes with more than one child.
The work can degenerate to $\tilde{O}(n^2)$ if we copy the best-decision arrays at the divergences, since we can end up with $O(n)$ leaves.
Sequentially, we can depth-first traverse the tree and compute the ``current'' best-decision array, and we only need to revert the array when backtracking.
However, this approach is inherently sequential.
To utilize the \ouralgo{} on a tree structure, we need to resolve the following two challenges: 1) how to efficiently identify the ready nodes;
and 2) how to efficiently maintain the best decision arrays for each node.

\begin{figure}[t]
  \centering
  \includegraphics[width=0.6\columnwidth]{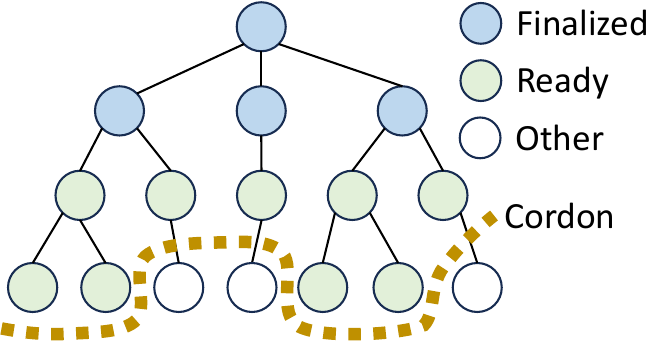}
  \caption{\ouralgo{} on a tree. Note that the sibling nodes must have the same status.
  \label{fig:tree}}
\end{figure}

%\subsubsection{Determine the Ready States}
\medskip
\myparagraph{Identifying the Ready States.}
Similar to the 1D case \cref{sec:post-office}, we maintain the best-decision array for each tree path. %, which stores the best decisions in its subtree.
Then we traverse the tree top-down, and we identify the ready states $v$ that can be finalized in this round, and compute them in parallel.
An illustration can be found in \cref{fig:tree}.
%that $v$ is tentative but $p(v)$ is finalized can be worked on in parallel.

Our high-level idea still follows the prefix-doubling technique, similar to the 1D case.
In the $t$-th doubling step we expand all nodes with $2^{t-1}\le d_v < 2^t$.
These nodes can be extracted by a range report shown in \cref{sec:tree-bb}. %tree depth range query from the frontier in step $t-1$.
We use prefix doubling and the checking process in \cref{sec:post-office} to decide the boundary that forms the cordon in the next round.
% Using prefix doubling guarantees work efficiency since the number of ready nodes is always asymptotically the same as checked nodes.
When checking the availability, we can use the HLD-based tree path query to find the minimum (highest) node on each path that is not available.
We will put sentinels on these nodes that block their subtrees.
The process stops when we find such nodes for all tree paths. %all nodes in the frontier are blocked.
In \cref{fig:tree}, these ready nodes are shown in green.
In the next round we can asynchronously work on the subtrees on the cordon in parallel.
We repeat this process until all nodes are finalized (correctly computed).

Here one difference to the 1D case is that the work cost of perfix doubling cannot be perfectly amortized.
In the 1D GLWS, if the prefix-doubling stops at step $t$,
we visit at least $2^{t-1}$ ready states and at most $2^{t-1}$ unready states,
thus the work to visit the unready states can be amortized.
However, in the tree case we are doing prefix-doubling by the depth of nodes.
The number of nodes in the last prefix-doubling step can be much larger than in the previous steps,
and the cost cannot be amortized.
The insight is that due to the prefix-doubling,
each node $v$ will be visited in at most $O(\log n)$ rounds.
Plus the $O(\log^2n)$ work of the range report,
the work in each round can be amortized to $O(h\log^3n)$, where $h$ is the frontier size.

%\subsubsection{Update best decision}
\myparagraph{Updating the Best-Decision Arrays.}
The most interesting part in this algorithm is how to maintain the best-decision arrays for all tree paths while achieving work efficiency and high parallelism.
%In this step the goal is to maintain the best decision array for each ready node.
Due to the tree structure, the best-decision arrays for different branches of the tree share some parts.
In total, there can be $O(n)$ paths with total sizes of $O(n^2)$.
The key challenge is to save the work and space by sharing parts of the arrays, while updating them highly in parallel.
%that we cannot afford the span proportional to the depth of the tree.

Consider the simple case when the ready nodes form a chain (the 1D case in \cref{sec:post-office}).
Here we use persistent BSTs to maintain the best-decision arrays on each node.
We first use \funcfont{UpdateBestChoice} in \cref{alg:dp-glws} to generate the best-decision array $B$ in the middle node of the chain,
and merge it with the old $B$ (the one stored at the node above this chain) using the similar technique in \cref{sec:post-office-concave}.
During this process we use path-copying to generate a new version the new array.
Then we work on the upper part and the lower part of the chain in parallel.
By this divide-and-conquer method, we can generate the best-decision array on each node of the chain
with $O(m\log^2 m)$ work and $O(\log^3m)$ span, where $m$ is the length of the chain.

In the general case, the structure of ready states can be arbitrary.
To achieve work-efficiency and high parallelism, our solution is in a ``BFS-style'' algorithm that utilizes the properties of HLD (see \cref{sec:tree-bb}). 
For all ready nodes in this round, we extend the heavy chain that is directly connected to the finalized nodes.
Since the heavy chain will not diverge, the approach is the same as the 1D case except for additional persistence, with work proportional to the total number of nodes and polylogarithmic span.
Once we finish updating the heavy chain, we will in parallel work on the light children of the nodes on the heavy chain we just proceeded.
The overall structure is similar to a BFS with heavy edges with weight 0 and light edges with weight 1.
Since each node only appears in one heavy chain, the work is still proportional to the number of ready nodes.
%then we follow the light edges that is linked to this heavy chain to work on subtrees in parallel.
We can also achieve high parallelism due to the fact that there can be at most $O(\log n)$ heavy chains and light edges from the root to any node $v$.
Hence, we can finish updating all paths in a logarithmic number of steps per round, which guarantees both work-efficiency and high parallelism.
%the span is polylogarithmic for each round.
Here we assume we build the HLD for the entire tree $T$ at the beginning, but we can also build the HLD for the ready states locally for each round.

Combining all pieces together, 
in each round we can determine the ready states and maintain the best-decision arrays
with work proportional to the number of ready states and polylogarithmic span.
We hence have the following theorem:
\begin{theorem}
  \ouralgo{} solves \treelws{} in $O(n\log^3n)$ work and $O(k\log^4n)$ span,
  where $k$ is the longest path in the best decision dependency graph.
\end{theorem}

\subsection{Parallel $k$-GLWS}\label{sec:k-lws}

Another well-known variant of GLWS is to limit the output that contains a fixed given number of $k$ clusters in the output~\cite{wang2011ckmeans,yang2022dp}.
Here we refer to it as the $k$-GLWS problem.
%Similar to LWS, the $k$-clustering problem is to divide a sequence into \emph{exactly} $k$ subsequences and minimize the total cost.
Formally, let $D[i,k']$ be the minimum cost for the first $i$ elements in $k'$ clusters, and the DP recurrence is:
\begin{equation*}
  D[i,k']=\min_{j<i}{D[j,k'-1]+w(j,i)}
\end{equation*}
where $w(j,i)$ is the cost of forming a cluster containing elements indexed from $j+1$ to $i$, and the boundary case $D[0,0]=0$ and $D[i,0]=+\infty$ for $i>0$.
Directly solving this recurrence takes $O(kn^2)$ work.
When the cost function $w$ is convex (which happens in many practical settings), the computation of each column in the DP table is a static matrix-searching problem, i.e., for a totally monotone matrix $A$ where $A[j,i]=D[j,k'-1]+w(j,i)$,
we want to compute the minimum element in each column of $A$.
Theoretically this problem can be solved in $O(n)$ work by the SMAWK algorithm~\cite{aggarwal1987geometric}, but this algorithm is quite complicated and inherently sequential.
Practically, there exists a simple divide-and-conquer algorithm with $O(n\log n)$ work~\cite{apostolico1990efficient}, which is similar to the function \funcfont{FindIntervals} in \cref{alg:dp-glws}.
This algorithm first computes the minimum element in the $(n/2)$-th column by enumerating all elements in this column, and recurse on two sides.
Due to the monotonicity, the minimum element in the $(n/2)$-th column limits the searches on both sides and guarantees the search ranges shrink by a half.
Hence, the work spent on each recursive level is $O(n)$, yielding $O(n\log n)$ total work for a recursive structure with $\log n$ levels.
By parallelizing the divide-and-conquer and using parallel reduce to find the minimum element (with $O(\log n)$ span), the total span is $O(\log^2 n)$.

We now show that when applying \ouralgo{} to this sequential algorithm, the $k'$-th frontier contains all states $D[\cdot,k']$.
We can see that in the first round, states $D[\cdot,1]$ are ready.
Since all states $D[\cdot,2]$ depend on some state from $D[\cdot,1]$, we will put \sentinel{s} on all states $D[\cdot,2]$ and they thus block all later states.
Then we can inductively show that this applies to all rounds, so finishing this computation requires $k$ rounds.
Then computing all states $D[\cdot,k']$ in each round using the aforementioned algorithm requires $O(n\log n)$ work and $O(\log^2 n)$ span.
Hence, the entire algorithm has $O(kn\log n)$ work and $O(k\log^2n)$ span.
In this problem, $k$ is also the depth of the DP DAG, so this algorithm is a perfect parallelization of the classic sequential algorithm.

\hide{
Here we show an parallel algorithm that can be directly derived by the \ouralgo{} framework.
In the beginning we finalize all $D[i,0]$ by the boundary condition.
By the \ouralgo{}, in each round we need to determine the ready states,
and use their DP values to update other tentative states.
As $D[*,k+1]$ relies on $D[*,k]$, we will put sentinels on column $k+1$,
thus the ready states are actually the column $D[*,k]$.
Then we try to use the $D[*,k]$ to relax $D[*,k+1]$.
We can use the similar divide-and-conquer technique as the \funcfont{FindIntervals} function in \cref{alg:dp-glws}
to achieve $O(n\log n)$ work and $O(\log^2n)$ span.
The whole algorithm takes $O(kn\log n)$ work and $O(k\log^2n)$ span.
In this problem $k$ is also the depth of the DP DAG, so this algorithm is a perfect parallelization.
} 

\subsection{Optimal Binary Search Tree (OBST)}
\label{sec:obst}
(Static) OBST is one of the earliest examples of DM optimization.
Given an array of frequency $a_{0..2n}$, it computes the recurrence 
\begin{equation}\label{eqn:obst}
  D[i,j]=\min_{j\le k<i}{D[j,k]+D[k+1,i]+w(i,j)}
\end{equation}
where $w(i,j)=\sum_{k'=2(i-1)}^{2j}a_{k'}$, and returns $D[1,n]$.
\citet{knuth1971optimum} first showed that computing this recurrence only needs $O(n^2)$ work, and later \citet{yao1980efficient} showed that this algorithm applies to any convex function $w(i,j)$.
Here let the best decision of a state $D[i,j]$ be the index $k$ that minimizes $D[i,j]$ in \cref{eqn:obst}.
In this algorithm, $D[i,j]$ depends on $D[i,j-1]$ (let $l$ be its best decision), $D[i+1,j]$ (let $r$ be its best decision), $D[i,l..r]$, and $D[l..r,j]$.
When applying \ouralgo{}, due to the dependence from $D[i,j]$ to $D[i,j-1]$ and $D[i+1,j]$, the $\delta$-th frontier contains the states $D[i,i+\delta]$.
Hence, although it results in optimal parallelization to the standard sequential algorithm, the algorithm requires $n-1$ rounds and thus has $O(n\log n)$ span. 
Achieving $o(n)$ span may need new insights to redesign the dependencies. 
}

\section{Experiments}\label{sec:exp}

To demonstrate the practicability of our new algorithms,
we designed experiments for LCS and convex GLWS. 
We implemented our parallel LCS algorithm and parallel GLWS algorithm in C++
using ParlayLib \cite{blelloch2020parlaylib} to support fork-join parallelism and some parallel primitives (e.g., reduce).
Our tests use a 96-core (192 hyperthreads) machine with four Intel Xeon Gold 6252 CPUs and 1.5 TB of main memory. 
We implemented two of our new algorithms (LCS and GLWS) as proofs-of-concept to show that our new techniques and algorithms remain practical due to work efficiency, which is also the main motivation of our paper.
We release our code at~\cite{dpdpcode}, which gives more details of our experiments.

\myparagraph{Parallel LCS.} While as one of the classic algorithmic problem, LCS is widely studied in parallel, the existing parallel implementations we know of
\cite{lu1994parallel,tchendji2020efficient,babu1997parallel,xu2005fast,apostolico1990efficient,chen2006fast,vahidi2023parallel}
do not take advantage of the sparsification. 
Most of them parallelizes the $\Theta(nm)$ algorithm for two strings with sizes $n$ and $m$,
and the experimental studies in these papers focus on input sizes up to $10^5$. 
They cannot process much larger instances in as the (sequential) sparse LCS algorithms (e.g.,~\cite{apostolico1987longest}) as discussed in \cref{sec:lis}.
Hence, in our experiments, our baseline is our implementation of the sparse LCS in~\cite{apostolico1987longest}, which can process our input instances with large size ($10^8$) and small number of effective edges $\totarrows$ (number of pairs $(i,j)$ such that $A[i]=B[j]$).

We test two random strings $A[1..n]$ and $B[1..n]$ with length $n=10^8$,
while controlling $\totarrows$ and $k$ (the LCS length). 
The pre-processing time to find all matching pairs is not counted into the running time.
\cref{fig:exp-lcs} shows the results when $\totarrows=10^8$ and $\totarrows=10^9$.

Overall, our algorithm has up to 30$\times$ speed up than the sequential version. Since the span of the algorithm is proportional to the LCS length $k$, the running time increases when $k$ increases. For $L=10^8$, the parallel running time stays competitive to the sequential running time until the extreme case $k=10^8$ (i.e., the two sequences are exactly the same). Since our algorithm has $O(\totarrows\log n)$ work and $O(k\log n)$ span, there is no parallelism when $\totarrows=k=10^8$. 
For $L=10^9$, where the total work is larger, our parallel algorithm is always faster than the sequential algorithm regardless of the value of $k$, and always achieves good parallelism. 

\myparagraph{Parallel GLWS.} For GLWS, we use the setting of the post-office problem described in \cref{sec:post-office}. 
We compare our parallel algorithm with the sequential solution in \cref{sec:post-office:prelim}, as well as our own algorithm running on one core.
We generate random data for $n=10^8$ and $10^9$, and 
%We use different values of $c$ (cost to build a post office) to control the output size $k$,
use different weight functions to control the output size $k$,
which is the number of post offices in the solution.
\cref{fig:exp-glws} shows the result on different $n$ and $k$.
The time for sequential algorithm does not change significantly,
because it has $O(n\log n)$ work, which is independent of $k$.
For our algorithm, the running time varies with $k$ due to the $O(k\log^2n)$ span. 
When $k$ is small, our algorithm is 20$\times$ faster than the sequential algorithm
and achieves 30--40$\times$ self-relative speedup. 
Our parallel algorithm is faster than the sequential algorithm for $k<10^4$ when $n=10^8$, and is faster than the sequential algorithm until $k\approx 10^5$ when $n=10^9$. 

\begin{figure}[t]
\centering
\begin{minipage}[b]{1.0\columnwidth}
    \centering
    \includegraphics[width=\columnwidth]{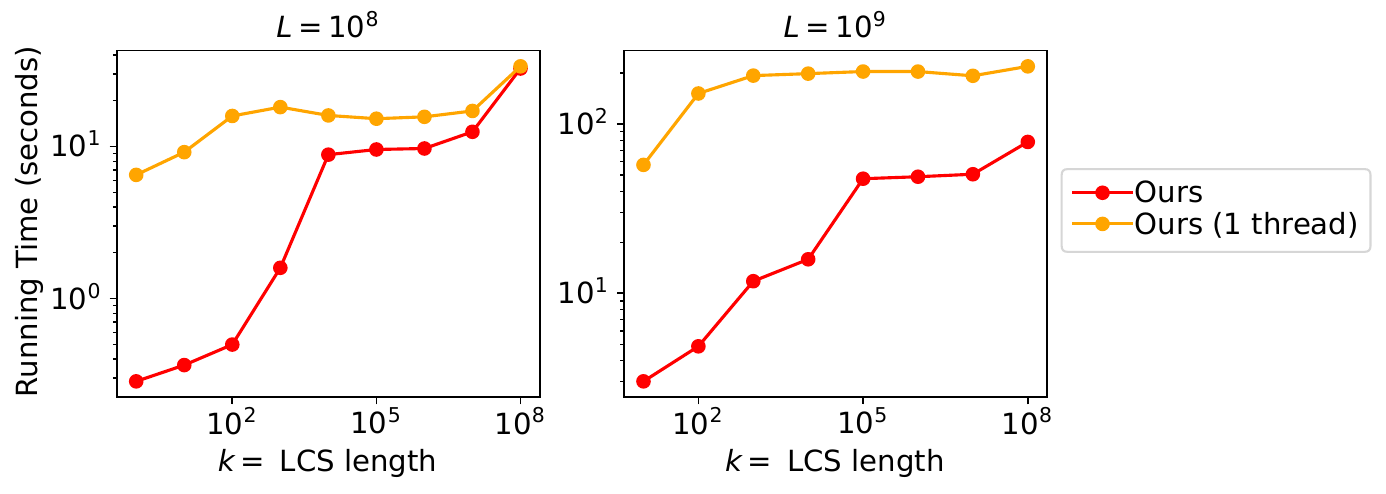}
    \caption{Running time of our parallel LCS algorithm (in \cref{sec:lis}).}
    \label{fig:exp-lcs}
\end{minipage}
\hfill
\begin{minipage}[b]{1.0\columnwidth}
    \centering
    \includegraphics[width=\columnwidth]{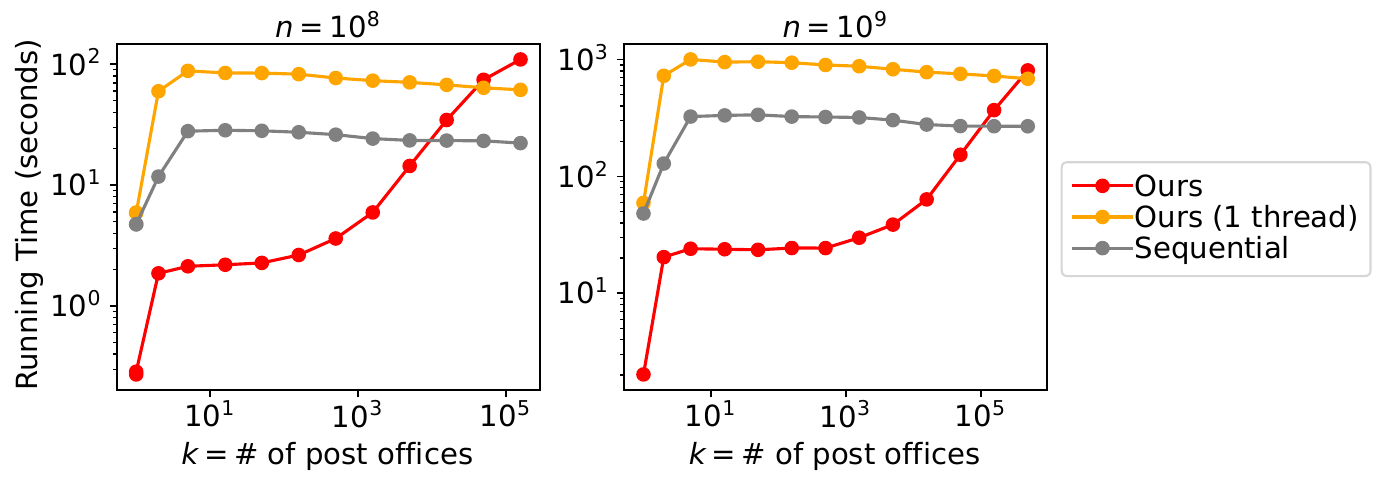}
    \caption{Running time of our parallel convex-GLWS algorithm (\cref{alg:dp-glws}).}
    \label{fig:exp-glws}
\end{minipage}
    %  \caption{Three simple graphs}
    %  \label{fig:exp}
\end{figure}

\section{Conclusion and Future Work}\label{sec:conclusion}

We systematically studied general approaches to parallelize classic sequential dynamic programming algorithms, 
particularly those with non-trivial optimizations such as decision monotonicity and sparsification.
We showed a novel framework, the \ouralgo{}, and apply it to different DP recurrences.
Theoretically, we gave the concept of optimal parallelism and perfect parallelism of a sequential algorithm, 
and showed that with a careful design, we can achieve optimal parallelism for the classic sequential DP algorithms in a (nearly) work-efficient manner, and perfect parallelism for some instances.
Practically, we show that our carefully-designed techniques do not include much overhead, and can outperform the original sequential version in a wide variety of cases.

We believe that the techniques in this paper opens a list of interesting questions.
First, many of the new parallel algorithms are nearly work-efficient---we pick the most practical sequential algorithms
for each problem, but they can be off the best work bound by up to an $O(\log n)$ factor.
It is theoretically interesting to ask if we can match the best work bound in parallel. 
Second, among all these classic algorithms we looked at, one problem/algorithm that 
cannot be directly solved by the \ouralgo{} is the RNA Secondary Structure~\cite{eppstein1988speeding}. 
%we do not know how to efficiently parallelize is the RNA Secondary Structure~\cite{eppstein1988speeding}.
Here, we may have $2^n$ different paths in a DP DAG, so applying \ouralgo{} efficiently may need some complicated techniques.
Finally, we show how to faithfully parallelize the sequential DP algorithms. We are aware of other approaches~\cite{lu1994parallel,tchendji2020efficient,xu2005fast,apostolico1990efficient,krusche2010new,cao2023nearly} for LIS/LCS
that can achieve stronger worst-case span bounds using divide-and-conquer. 
%Due to their divide-and-conquer nature, they can achieve stronger worst-case span bounds.
Hence, an interesting direction is to see if we can redesign other DP algorithms in a similar form to achieve better worst-case span.

\section*{Acknowledgement}
This work is supported by NSF grants CCF-2103483, IIS-2227669, NSF CAREER Awards CCF-2238358 and CCF-2339310, the UCR Regents Faculty Development Award, and the Google Research Scholar Program.

\bibliographystyle{ACM-Reference-Format}
\ifconference{\balance}
\bibliography{bib/strings, bib/main, local}

\iffullversion{
\appendix

\section{Additional Details for Parallel Optimal Alphabetic Trees (OAT)}\label{sec:app-oat}

This section supplements additional details for the parallel OAT algorithm in \cref{sec:oat}.

\subsection{Sequential Algorithms for OAT}

Hu-Tucker algorithm \cite{hu1971optimal} was the first $O(n\log n)$ algorithm for OAT,
and a simplified version Garsia-Wachs algorithm \cite{garsia1977new} was proposed later.
Both algorithms have two phases.
In the first phase, a certain tree called the $l$-tree (short for level-tree) is constructed from the input sequence.
Hu-Tucker algorithm and Garsia-Wachs algorithm differ in the manner to find the $l$-tree, but the final $l$-tree is the same.
In the second phase, the OAT can be constructed from the $l$-tree,
and each item in the OAT is at the same level as in the $l$-tree.

The key problem is how to compute the $l$-tree.
Here we describe the first phase of Garsia-Wachs,
as it is easier to understand than Hu-Tucker and the idea is used in the parallel algorithm by \citet{larmore1993parallel}.
Let the input weight sequence be $a_{1..n}$.
We denote $a_i+a_{i+1}$ be the $i$'th 2-sum, for $1\le i < n$.
A pair of consecutive elements $(a_i,a_{i+1})$ is said to be locally minimal if the $i$'th 2-sum is a local minimum in the sequence of 2-sums.
% For simplicity, each element's weight will also be its name\yan{what does this mean?}.
The Garsia-Wachs algorithm repeatedly performs the following steps until there is only one element in the sequence:
\begin{enumerate}
  \item Find the left-most locally minimal pair $(a_i,a_{i+1})$.
  \item Combine $a_i$ and $a_{i+1}$.
  Make a new node $x$ to be the parent of $a_i$ and $a_{i+1}$ in the $l$-tree,
  and the weight of $x$ is $a_i+a_{i+1}$.
  \item Remove $a_i$ and $a_{i+1}$ from the sequence.
  Insert $x$ before the first element $a_j$ where $j>i$ and $a_j\ge x$.
  If such $a_j$ does not exist, insert $x$ at the end of the sequence.
\end{enumerate}
Note that each newly generated node represents a tree, with both children as the two trees corresponding to the two nodes it merges in step (2).
%the sequence becomes a forest and each element is the root of a subtree in the final $l$-tree.
%The parent links in the forest are introduced in step (2).
At the end, there will be only one element in the sequence, which is the final output of the $l$-tree.
%At last, the remaining element in the sequence is the root of the $l$-tree.
An $l$-tree can be easily converted to an OAT in parallel with $O(n)$ work and $O(\log n)$ span, and we refer the audience to~\cite{larmore1993parallel} for more details.

\subsection{The Parallel Algorithm by \citet{larmore1993parallel}}

Here we focus on the first phase on how to construct the $l$-tree in parallel.
%Here we provide a detailed description of Lawrence's parallel OAT algorithm.
Larmore et al.~observe that we can pick any locally minimal pair instead of the left-most one as in Garsia-Wachs,
which does not affect the resulting $l$-tree.
Hence, we can in parallel process all possible locally minimal pairs in one round.
However, in the worst case the number of rounds is still linear.
For example, if $a_{1..n}$ is in increasing order, then only one locally minimal pair can be processed in one round.

To overcome this issue, Larmore et al.~further proposed the concept of ``valleys''.
A valley $\alpha=\vall(a_i)$ is the \emph{largest} contiguous subsegment of $a$ that 1) contains $a_i$, and 2) contains no item larger than $a_i$.
Hence, $a_i$'s valley contains the elements from $a_i$'s subtree in the Cartesian tree of $a$.
Larmore et al.~showed that several disjoint valleys in the sequence can be processed in parallel.
To understand this, let us consider a valley $\alpha$.
Let $\Delta_\alpha$ be the parent of valley $\alpha$ in the Cartesian tree.
Since $\alpha$ is maximal, the locally minimal pairs from this range will not interact with the outside elements.
Hence, we can consider valley $\alpha$ as an independent task, and repeatedly find and combine the locally minimal pairs inside $\alpha$.
The only difference is that if the combined element $x=a_i+a_{i+1}>\Delta_\alpha$,
we mark it as $x^*$ and put it in a separate queue.
After the parallel processing of disjoint valleys, we collect all marked elements and insert them into the right place.

In total, there can be at most $n$ overlapping valleys.
Larmore et al.'s algorithm uses one special type of valley: the 1-valleys.
A 1-valley is defined to be a valley $\alpha$ such that for any node $v$ in the subtree of $\alpha$ in the Cartesian tree,
if $v$ has two children, at least one of the children is a leaf.
We can see that in the sequence the maximal 1-valleys are not overlapping, so we can process them in parallel.
%Also, we can find some number of maximal 1-valleys in the sequence.
They also proved that, in each round if we process all maximal 1-valleys and reinsert them as in the sequential order,
the number of maximal 1-valleys in the remaining sequence decreases by at least a half. %is no more than half as in the original sequence.
Thus, the whole algorithm will finish in $O(\log n)$ rounds.

Before we show how to process a 1-valley, we introduce some definitions.
Let $p(u)$ is the parent of $u$ in the Cartesian tree of $\alpha$.
We define a 1-valley $\alpha$ to be a regular valley if for $u<v$ are two leaves in the Cartesian tree of $\alpha$, then $p(u)<p(v)$.
A sorted regular valley is defined as a regular valley with the minimum element in the first one.
Larmore et al.\ showed that a 1-valley can be transformed to a sorted regular valley in $O(m)$ work and $O(\log m)$ span, where $m$ is the size of this 1-valley.
A set of items $S$ is defined to be lf-closed (short for leaf-father-closed) if for any $a_i\in S$ and $a_i$ is a leaf in the Cartesian tree of $\alpha$,
then $p(a_i)$ is also in $S$.
The weight of $S$ is defined as the sum of the weights of all items in it.
We define $W_k$ as the minimum weighted lf-closed set with $k$ items in it.
Larmore et al.\ showed that all $W_k$ can be computed in $O(m\log m)$ work and $O(\log m)$ span.

Now consider a sorted regular valley $\alpha$ with length $m$.
Our goal is to generate a sequence of forests $\forest_{0..m'-1}$ ($m'<m$) so that each of them corresponds to a subset of subtrees in the $l$-tree.
Here %we do not include the single nodes in the forest, thus 
$\forest_0$ is empty.
The last forest $\forest_{m-1}$, if $m=m'$, corresponds to the $l$-tree if we merge all elements in this valley.
However, this process may not end here---as mentioned above, we will stop if the weights of subtrees exceed $\Delta_\alpha$.
Hence, we can have our final state $\forest_{m'-1}$ with $m'<m$.
To compute $\forest_{i}$, we will enumerate all $\forest_{j}$ for $j<i$, and find the best (minimum) transition from them.
Here a transition means to build an additional level in the $l$-tree, and the cost $w(j,i)=W_{2i-j}$.
% \yan{add the details about what $W$ is and how to compute it.} % in the previous paragraph
Larmore et al. showed that the cost function $w$ is convex, so computing $\forest_i$ is exactly a convex GLWS problem.
If we use \cref{alg:dp-glws} to solve this convex GLWS problem, since each decision $\best[i]=j$ will add another level to the trees in $\forest_i$ from $\forest_j$, the effective depth is upper bounded by the overall $l$-tree height $h$.
Hence, the work and span for each 1-valley subproblem is $O(m\log m)$ work and $O(h\log^2 m)$ span, where $m$ is the subproblem size.
Since the total size of 1-valleys in a recursive round is $O(n)$, the work and span for one round are $O(n\log n)$ work and $O(h\log^2 n)$, respectively.
Multiplying this by $O(\log n)$, the number of recursive rounds, gives the cost bounds in \cref{theorem:oat}.

\hide{
Each $\forest_i$ ($i>0$) adds a new root $p_i$ by combining two existing nodes (either a root from $\forest_{i-1}$ or a single node from $\alpha$).
Lawrence proves the following lemma:
\begin{lemma}\label{lemma-forest}
  Let $0<i<m$. Then there exists $0\le j < i$ that:
  (1) the roots of $\forest_i$ are $p_{j+1..i}$;
  (2) the leaves of $\forest_i$ are $W_{2i-j}$;
\end{lemma}
Note that from $\forest_j$ to $\forest_i$ we added one level,
thus the weight of $\forest_j$ is the weight of $\forest_i$ plus the weight of $W_{2i-j}$.
To find the minimum weighted forests, we can reduce to a LWS problem on $0..m-1$ by defining $w(j,i)=W_{2i-j}$.
Larmore et al. showed that the cost function $w$ is convex, so computing $\forest_i$ is exactly a convex GLWS problem.
Since each decision $\best[i]=j$ will add another level to the trees in $\forest_i$ from $\forest_j$, 
After solving this convex LWS, the final forest can be recovered in $O(m\log m)$ work and $O(\log m)$ span.
As \cref{lemma-forest} shows, each $\best[i]=j$ brings an extra level in $\forest_i$.
So the height of the final forest is the same as the longest path in the dependency graph.
As mentioned above, the final forest is a subgraph in the final $l$-tree, and the OAT has the same height of the $l$-tree.
So we have $k$, the length of the longest dependency path, is upper bounded by the final OAT height $h$.}

\subsection{Proof of Lemma~\ref{lem:oat-height-lemma}}

Here we provide the proof of Lemma~\ref{lem:oat-height-lemma}. We first show the following lemma. 

\begin{lemma}\label{lem:lem-oat-height-lemma}
  Let the weight of a subtree in an OAT as the total weight of all leaves in this subtree. 
  In an OAT, the subtree weight grows by at least twice for every three levels.
\end{lemma}

\begin{proof}
Here we denote $p(v)$ as the parent node of $v$ in the output OAT $T$ and $w(v)$ as the sum of the leaf weights in the subtree of node $v$.
%In the OAT, we denote $w(v)$ as the sum of the leaf weights in the subtree of node $v$.
%So $w(v)$ is the sum of weights of $v$'s children.
We will show for any node $a$ in the OAT, the great-grandfather of $a$ (if exists) must have weight no less than $2w(a)$.
In the optimal alphabetic tree $T$, let node $b$ be $a$'s sibling and node $c$ be $a$'s parent.
WLOG we assume node $a$ is the left child of node $c$.
Then, let $d$ be $c$'s sibling, $e$ as $c$'s parent, and $f$ as $e$'s parent.
Now the lemma is equivalent to $w(f)\ge 2w(a)$.

\begin{figure}[h]
  \centering
  \includegraphics[width=1.0\columnwidth]{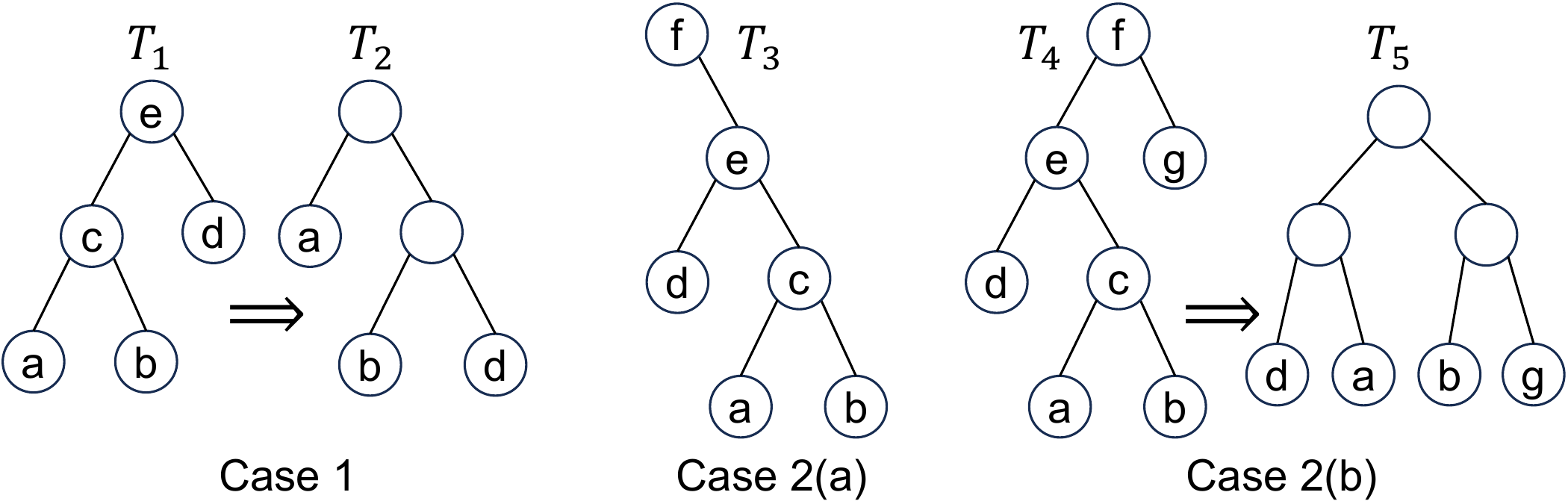}
  \caption{Illustrations for analyzing the OAT Height.  The three cases are used in the proof of \cref{lem:lem-oat-height-lemma}.
  \label{fig:oat-height}}
\end{figure}

We consider the first case when $c$ is $e$'s left child (case 1 in \cref{fig:oat-height}).
In this case, we must have $w(d)\ge w(a)$, since otherwise a right rotation will decrement the total cost of the tree by $w(d)-w(a)$, violating that $T$ is an OAT.
Hence, $w(f)\ge w(e)\ge w(d)+w(a)\ge 2w(a)$ in this case.

The second case is when $c$ is $e$'s right child.
There are two sub-cases.
First, if $e$ is $f$'s right child (case 2(a) in \cref{fig:oat-height}), then $c$, $e$, and $f$ form another case 1.  Hence, we have $w(f)\ge 2w(c)\ge 2w(a)$.
Second, if $e$ is $f$'s left child (see $T_4$ in case 2(b) in \cref{fig:oat-height}). We can double-rotate and get another valid alphabetic tree $T_5$.
As $T_4$ is the optimal alphabetic tree, we must have $2w(d)+3w(a)+3w(b)+w(g)\le 2(w(d)+w(a)+w(b)+w(g))$,
which leads to $w(g)\ge w(a)+w(b)$.
So $w(f)=w(d)+w(a)+w(b)+w(g)\ge w(a)+w(g)\ge 2w(a)$.

From all three cases above, we show $w(f)\ge 2w(a)$, which proves the lemma that the weight doubles for every three steps up the OAT.
\end{proof}

With integer weights in word size $W$, the weight of the root is at most $O(W)$, and the weight of each leaf is at least 1.
In this case, the number of levels between them is at most $O(\log W)$. 
This proves \cref{lem:oat-height-lemma}.

\hide{\noindent Case 1: if $c$ is $e$'s left child (see $T_1$ in \cref{fig:oat-height}).
We can do a right rotate to get $T_2$, which is also a valid alphabetic tree.
As $T_1$ is the optimal alphabetic tree, we must have $2w(a)+2w(b)+w(d)\le w(a)+2w(b)+2w(d)$, which leads to $w(a)\le w(d)$.
So $w(e)=w(a)+w(b)+w(d)\ge 2w(a)$, and the parent of $e$ ($a$'s great-grandfather) must have weight at least $2w(a)$.

\noindent Case 2: if $c$ is $e$'s right child. There are two sub-cases.
(a) If $e$ is $f$'s right child, then $c,e,f$ forms case 1 (symetric). So we have $w(f)\ge 2w(c)\ge 2w(a)$.
(b) If $e$ is $f$'s left child (see $T_4$ in \cref{fig:oat-height}). We can do a double-rotate and get another valid alphabetic tree $T_5$.
As $T_4$ is the optimal alphabetic tree, we must have $2w(d)+3w(a)+3w(b)+w(g)\le 2(w(d)+w(a)+w(b)+w(g))$,
which leads to $w(g)\ge w(a)+w(b)$.
So $w(f)=w(d)+w(a)+w(b)+w(g)\ge w(a)+w(g)\ge 2w(a)$.

So the weight of the root of OAT is at least $2^{h/3}$.
If $W$ is the word size, we must have $2^{h/3}<W$, which leads to $h=O(\log W)$.
}

}

\end{document}
\endinput
%%
%% End of file `main.tex'.